\documentclass[
reprint,
superscriptaddress,
amsmath,amssymb,
aps,
pra
]{revtex4-2}

\usepackage[utf8]{inputenc}
\usepackage[T1]{fontenc}

\usepackage{csquotes}
\usepackage{comment}
\usepackage{lipsum}
\usepackage{calc}
\usepackage{soul}

\usepackage{amsmath}
\usepackage{amssymb}
\usepackage{amsthm}
\usepackage{cancel}
\usepackage{braket}
\usepackage{multirow}
\usepackage{mathtools}
\usepackage{nicematrix}

\newtheorem{theorem}{Theorem}
\newtheorem{definition}[theorem]{Definition}
\newtheorem{proposition}[theorem]{Proposition}
\newtheorem{corollary}[theorem]{Corollary}
\newtheorem{remark}[theorem]{Remark}

\usepackage[caption=false]{subfig}
\usepackage{placeins}
\usepackage[svgnames]{xcolor}

\usepackage{graphicx}
\usepackage[export]{adjustbox}

\usepackage{hyperref}
\hypersetup{colorlinks=true,allcolors=blue}

\DeclareMathSymbol{\mathleftquote}{\mathord}{operators}{'134}
\DeclareMathSymbol{\mathrightquote}{\mathord}{operators}{'42}

\renewcommand{\subset}{\subseteq}
\renewcommand{\phi}{\varphi}
\renewcommand{\mod}{\mathrel\mathrm{mod}}
\newcommand{\tabspace}{\quad}
\newcommand{\tabbspace}{\qquad}
\newcommand{\mathand}{\text{and}}
\newcommand{\mathor}{\text{or}}

\newcommand{\LHS}{\mathrm{LHS}}
\newcommand{\RHS}{\mathrm{RHS}}
\newcommand{\MHS}{\mathrm{MHS}}
\newcommand{\IN}{\mathrm{IN}}
\newcommand{\OUT}{\mathrm{OUT}}
\newcommand{\complex}{\mathbb{C}}
\newcommand{\LC}{\mathrm{LC}}
\newcommand{\bigstrut}{\vphantom{\big|}}
\newcommand{\LCboxed}{\boxed{\LC\bigstrut}}
\newcommand{\CZ}{CZ}
\newcommand{\SWAP}{\mathrm{SWAP}}
\newcommand{\CNOT}{\mathrm{CNOT}}

\newcommand{\paulis}{\mathcal{P}}
\newcommand{\cliffords}{Cl}
\newcommand{\symplecticgroup}{\mathrm{Sp}}
\newcommand{\centerstab}{\mathrm{center}}
\newcommand{\logicalX}{\overline{X}}
\newcommand{\logicalZ}{\overline{Z}}
\newcommand{\enum}{\mathrm{enum}}
\newcommand{\graph}{\mathrm{graph}}
\newcommand{\supp}{\mathrm{supp}}
\newcommand{\meas}{\mathrm{meas}}
\newcommand{\transv}{\mathrm{transv}}
\newcommand{\identity}{\mathrm{id}}
\newcommand{\ad}{\mathrm{ad}}
\newcommand{\rot}{\mathrm{rot}}
\newcommand{\zxcalculus}{ZX-calculus}
\newcommand{\nkd}[3]{$[\![#1,#2,#3]\!]$}
\newcommand{\supplmat}{SM}

\NiceMatrixOptions{
  matrix/columns-type = c c c c | c c c c
}

\newcommand{\vcenterinclude}[1]{\includegraphics[valign=c]{#1}}

\newcommand{\beginsupplementary}{
    \FloatBarrier
    \clearpage
    \onecolumngrid
    \appendix
    \renewcommand{\appendixname}{}
    \setcounter{theorem}{0}
}

\begin{document}

\preprint{APS/123-QED}
\title{Chutes and Ladders: Dynamical Automorphisms via the ZX-Calculus}
\author{Alexander Frei}
\email{alexander.frei@uwaterloo.ca}
\thanks{Joint first author.}
\affiliation{Institute for Quantum Computing, University of Waterloo, Waterloo, Ontario N2L 3G1, Canada}
\affiliation{Department of Pure Mathematics, University of Waterloo, Waterloo, Ontario N2L 3G1, Canada}
\affiliation{Perimeter Institute for Theoretical Physics, Waterloo, Ontario N2L 2Y5, Canada}
\author{Sascha Zakaib--Bernier}
\email{szakaibbernier@uwaterloo.ca}
\thanks{Joint first author.}
\affiliation{Institute for Quantum Computing, University of Waterloo, Waterloo, Ontario N2L 3G1, Canada}%
\affiliation{Department of Physics and Astronomy, University of Waterloo, Waterloo, Ontario N2L 3G1, Canada}%
\affiliation{Perimeter Institute for Theoretical Physics, Waterloo, Ontario N2L 2Y5, Canada}%
\author{Zachary Mann}
\affiliation{Institute for Quantum Computing, University of Waterloo, Waterloo, Ontario N2L 3G1, Canada}
\affiliation{Department of Physics and Astronomy, University of Waterloo, Waterloo, Ontario N2L 3G1, Canada}
\affiliation{Department of Physics, California Institute of Technology, Pasadena, California 91125, USA}
\author{Michael Vasmer}
\affiliation{Centre Inria de Paris, 48 rue Barrault, Paris 75013, France}
\affiliation{Institute for Quantum Computing, University of Waterloo, Waterloo, Ontario N2L 3G1, Canada}
\affiliation{Perimeter Institute for Theoretical Physics, Waterloo, Ontario N2L 2Y5, Canada}
\author{Victor V. Albert}
\affiliation{Joint Center for Quantum Information and Computer Science,
NIST/University of Maryland, College Park, Maryland 20742, USA}
\date{\today}

\begin{abstract}
The ZX-calculus is a powerful graphical language for manipulating quantum circuits, which has recently found many applications in quantum error correction.
We extend this language to handle Floquet and other dynamical stabilizer codes via the connection between measurement-based code switching and gauge fixing \cite{vuillot_code_2019}.
We combine gauge-fixing steps to implement a closed loop in the space of stabilizer codes, returning to the original codespace up to a logical Clifford gate. These measurement-based paths in the space of stabilizer codes can be viewed as shortcuts, or \enquote{chutes and ladders}, relative to single-qubit Clifford operations and qubit permutations. This yields a machine-interpretable method for constructing dynamical automorphisms and facilitates the search for implementations of desired logical gates. As an example, we implement a logical phase gate via distance-preserving code switching for the seven-qubit code bare code \cite{muyuan2017bare}, which has no non-trivial logical Clifford gates based on single-qubit Clifford operations and qubit permutations \cite{sayginel_code_aut_2025}.
\end{abstract}

\maketitle

\textit{Introduction.---}
Quantum information promises computations that are intractable with classical computers and communication in a way that allows eavesdroppers to be detected~\cite{nielsen2010quantumcomputation}. 
Despite these promises, quantum technologies are still in their infancy. 
The main reason for this is the fragility of quantum information; slight perturbations due to the environment can destroy quantum advantages.
Quantum error correction (QEC) is an essential step toward scalable quantum computing, as it encodes logical information across many physical qubits in a subspace where errors can be detected and corrected.
Recent experiments have demonstrated QEC across various platforms including superconducting qubits~\cite{ofek2016extending, krinner2022realizing, acharya_suppressing_2023, acharya_quantum_2025, lacroix_scaling_2025}, trapped ions~\cite{ryan-anderson2021realizationrealtime,postler2024demonstration,paetznick2024demonstrationlogical,reichardt2024demonstrationquantum,ryan-anderson2024highfidelityfaulttolerant}, and neutral atoms \cite{muniz2025repeated,bluvstein_architectural_2025}. 

While there is still room for improvement on the hardware and QEC techniques, the next step is performing fault-tolerant quantum algorithms. 
The standard approach is to decompose an algorithm into a small number of quantum gates forming a universal set~\cite{nielsen2010quantumcomputation}, and then to design fault-tolerant implementations of these gates in a given QEC code.
Transversal gates~\cite{gottesman_stabilizer_1997} have a tensor-product structure and are therefore naturally fault tolerant, but no QEC code possesses a universal and transversal set of logical gates~\cite{eastin_restrictions_2009}.
As a result, for any QEC code, at least one logical gate must be implemented using a more complex fault-tolerant scheme, such as, e.g., lattice surgery~\cite{horsman2012surfacecode,Litinski2019gameofsurfacecodes,cohen2022lowoverheadfaulttolerant,ide2025faulttolerantlogical,swaroop2024universaladapters,he2025extractorsqldpc}, magic state distillation~\cite{knill2004faulttolerantpostselected,bravyi_universal_2005,bravyi2012magic,Haah2018codesprotocols,Litinski2019magicstate}, and code switching~\cite{paetznick_universal_2013,bombin2015gaugecolor,kubica2015universal,Bombin_2016,beverland2021cost}.
Even so, no consensus yet exists on how to best perform the logical operations necessary to compute any quantum algorithm fault-tolerantly.

The usual QEC paradigm is to encode the logical information into a single subspace (the codespace).
However, recent proposals have explored QEC techniques where the system is dynamically driven between different codespaces~\cite{hastings_dynamically_2021,davydova2023floquet,kesselring2024anyon,higgott2024constructions,bombin_unifying_2024}.
These dynamical codes enable QEC using lower-weight measurements, and can also be used to implement logical Clifford gates~\cite{davydova_quantum_2024}.

In this Letter, we develop a ZX-calculus--based formalism for dynamical automorphisms for the implementation of logical Clifford gates.
The ZX-calculus \cite{coecke_interacting_2011,kissinger_picturing_2024} is a graphical language used to reason about and beyond quantum circuits, and has recently found much success in the design of QEC protocols. However, it has not yet been applied to the construction of automorphisms of dynamical codes.
Complementing prior ZX-calculus applications that address the \textit{construction} of dynamical codes \cite{hastings_dynamically_2021,davydova2023floquet,bombin_unifying_2024,townsend-teague_floquetifying_2023,rodatz_floquetifying_2024,jacoby2026stairway}, we leverage ZX-calculus to design and implement their \textit{automorphisms}.
We expect that the machine-interpretability of this formalism will facilitate the search for new dynamical codes and logical gates.

The key ingredient of our approach is a ZX-calculus representation of gauge fixing~\cite{paetznick_universal_2013,bombin2015gaugecolor}, a technique to transfer encoded information between stabilizer codes via measurement-based shortcuts, which may fault-tolerantly implement high depth unitary circuits.
We give three examples of dynamical automorphisms via gauge fixing:
the logical Clifford group for a \([\![3,1]\!]\) toy code (for illustration),
an entangling logical $\CZ$ gate for the error-detecting \nkd{4}{2}{2} code, and a logical phase gate for the error-correcting \nkd{7}{1}{3} bare code \cite{muyuan2017bare}.

\textit{Background.---}
An $n$-qubit stabilizer code~\cite{gottesman_stabilizer_1997,calderbank1997quantum} is specified by an Abelian subgroup (the stabilizer group) of the $n$-qubit Pauli group that does not contain $-I$.
The number of encoded qubits is given by $k = n - m$, where $m$ is the number of independent generators of the stabilizer group.
Logical Pauli operators are those Pauli operators that commute with the stabilizer group but are not contained within it.
We denote a stabilizer code using the shorthand \nkd{n}{k}{d}, where $d$ is the code distance, defined as the minimum weight of a logical Pauli operator.
 
A stabilizer code is represented in the ZX-calculus by its encoding map. We take that to be a Clifford isometry between the logical qubits and physical qubits. This isometry is an example of a ZX-diagram that can be called a ZX-encoding graph~\cite{kissinger_phase_free_2022, khesin_universal_2025} (see Section \ref{sec:zx-calculus} of the Supplementary Material (SM)). 
\begin{definition}\label{def:ZX-encoding_graph}
    A ZX-encoding graph is a bipartite graph-like ZX-diagram, where the biadjacency matrix is full rank and single-qubit Clifford gates can be added to the input or output wires.
\end{definition}
The condition on the rank of the biadjacency matrix comes from the constraint that the ZX-diagram must be an isometry (see Section \ref{subsec:zx-enc_full_rank} of the SM). We point out that from now on, local Clifford (LC) gates refer to single-qubit Clifford gates.

Via map--state duality, it is possible to associate a given ZX-encoding graph with a single stabilizer state \cite{wu_zx-calculus_2024}. This state, in turn, corresponds to a cluster (a.k.a. graph) state, up to local Clifford equivalence~\cite{nest_graphical_2004}. This enables us to import results from graph states into the context of stabilizer codes. 

As a motivating example, we illustrate a ZX-encoding graph for a trivial 3-qubit code on the left-hand side of the equation below,
\begin{equation}\label{eq:stab_eq_graph_LC}
    \hspace{-4.5mm}
    \mathrel{\vcenter{\hbox{ \includegraphics{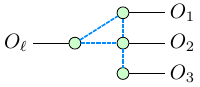}}}}
     \cong
     \left[
     \mathrel{\vcenter{\hbox{ \includegraphics[height=1.25cm]{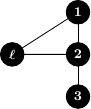}}}}, \{O_\ell, O_1, O_2, O_3 \}
     \right]\hspace{-1mm},
\end{equation}
where the logical input of the encoder is the left-most wire, the physical outputs on the three right-most wires, and where $O_i$ are LC gates.
The map--state duality is manifest on the right-hand side, where we draw the graph defining the corresponding graph state.

Each side of the above equivalence contains all information about the code.
The code stabilizers are obtained by restricting the graph-state stabilizers to those \textit{not} supported on any vertices that correspond to the code's input logical qubits.
Recall that the stabilizer generators of the graph state associated to a graph $G=(V,E)$ are given by $\{X_v \bigotimes_{u\in N(v)} Z_u, \hspace{1mm} \forall \hspace{1mm} v \in V \}$, where $N(v)$ denotes the neighbourhood of a vertex $v$.
Restricting the graph-state stabilizers in the above example yields the code's stabilizer group generated by $YYZ$ and $IZX$.
Logical Paulis can also be obtained.
For example, logical $X$ operator of the 3-qubit code corresponds to the graph state stabilizer acting as $X$ on the input vertex, i.e., $X_{\ell} = ZZI$.
Analogously, we have $Z_{\ell} = XZI$ and $Y_{\ell} = YII$.

Local Cliffords gates, generated by $\{ \sqrt{X}, H\}$, applied to graph states \cite{nest_graphical_2004} and graph-like ZX-diagrams \cite{duncan_graph_2009} correspond to local complementations of the underlying graph $G$.
A local complementation about a vertex complements the subgraph created by that vertex's neighbors.
We illustrate the equivalence between a local complementation and local Clifford gates via the following example:
\begin{equation}\label{eq:rule-LC}
    \mathrel{\vcenter{\hbox{
    \includegraphics{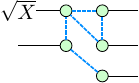}}}}
    \hspace{2mm} \cong \hspace{1mm}
    \mathrel{\vcenter{\hbox{
    \includegraphics{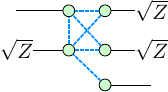}}}} \hspace{1mm},
\end{equation}
up to Pauli operators that we ignore from now on 
\footnote{Strictly speaking, the $\sqrt{X}$ in Eq.~\eqref{eq:rule-LC} should be $\sqrt{X}^{-1} = X \sqrt{X}$
(see Theorem \ref{thm:LC_equiv} of the SM).
Note however that Pauli operators are transversal gates, and thus come for free from a fault-tolerance perspective.
As such we may work in the symplectic representation (see the \supplmat\ section \ref{sec:notation_conventions}),
and from now on identify $\sqrt{X}=\sqrt{X}^{-1}$ and $\sqrt{Z}=\sqrt{Z}^{-1}$.}.

From equation \eqref{eq:rule-LC}, performing another local complementation on the upper right vertex, and a last one on the original vertex would produce $\sqrt{X}\sqrt{Z}\sqrt{X} = H$ gates on both vertices (up to a global phase).
This series of three alternating local complementations is called edge complementation, or pivoting \cite{bouchet_graphic_1988, duncan_pivoting_2014}. 
The equivalence between $H$ and pivoting allows us to move a Hadamard from a vertex to any other vertex connected by an edge (see Section \ref{sec:zx-calculus} of the SM). 

The ZX-encoding graph for a given code is not unique, but a normal form can be defined using that of stabilizer states \cite{hu_improved_2022}.
In this form, there can be only $\{I, \sqrt{Z}, H\}$ on open wires, and given a predetermined ordering of vertices, the Hadamard gates are sent to the highest numbered vertices via pivoting. 
One can convert a ZX-encoding graph into normal form by removing $\sqrt{X}$ gates via local complementation and by annihilating neighbouring Hadamard gates (see Section \ref{subsec:zx-enc_norm_form} of the SM).

\textit{Gauge fixing via the ZX-calculus.---}
Subsystem codes are stabilizer codes where we choose to relabel a subset of the logical qubits as gauge qubits, in which no logical information will be encoded~\cite{kribs2005unified,kribs2006operatorquantum,poulin_stabilizer_2005}. In this work, we consider subsystem codes with a single gauge qubit. 
Given the subsystem code's ZX-encoding graph, this translates to relabelling an input wire of a ZX-encoding graph to represent a gauge qubit. 
We illustrate this relabelling with the following subsystem code:
\begin{equation}
    \mathrel{\vcenter{\hbox{
\includegraphics{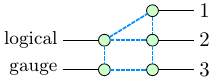}}}} .
\end{equation}
The gauge input has its own set of logical Pauli operators ${Z_g, X_g, Y_g}$ that are found by multiplying the vertex stabilizers of the associated graph state as before. 

Gauge fixing corresponds to putting the gauge qubit into a fixed state, which can be achieved through projective measurement of a logical gauge operator (potentially followed by a Pauli correction).
In the ZX-calculus, the three types of Pauli measurements on a qubit correspond to terminations of the qubit's open wire by ZX spiders.
\begin{definition}\label{def:gauge_fixing_zx}
    A gauge input in a subsystem code's ZX-encoding graph is gauged-fixed if its wire is terminated with an Z/X/Y eigenstate, respectively,
    $\includegraphics{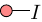} \hspace{2mm}, \hspace{2mm}
\includegraphics{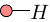} \hspace{2mm}, \hspace{2mm}
\includegraphics{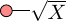} \hspace{2mm}$.
\end{definition}
\begin{figure*}[t!]
    \centering
    $\mathrel{\vcenter{\hbox{\includegraphics{ 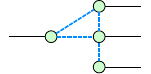}}}}
    \hspace{1mm} = \hspace{1mm}
    \raisebox{-6mm}{\includegraphics{ figures/gauge_fixing/Zgauge.pdf}}
    \hspace{-1.5mm} \left[ \hspace{-1mm} 
    \mathrel{\vcenter{\hbox{\includegraphics{ 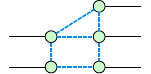}}}} \hspace{-1mm}\right]
    \hspace{1mm} = \hspace{1mm}
    \raisebox{-6mm}{\includegraphics{ figures/gauge_fixing/Xgauge.pdf}}
    \hspace{-1.5mm} \left[ \hspace{-1mm}
    \raisebox{-6.25mm}{\includegraphics{ 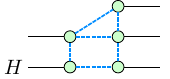}} \hspace{-1mm}\right]
    \hspace{1mm} \cong \hspace{1mm}
    \raisebox{-6mm}{\includegraphics{ figures/gauge_fixing/Ygauge.pdf}}
    \hspace{-1mm} \left[ \hspace{-1mm}
    \raisebox{-6.65mm}{\includegraphics{ 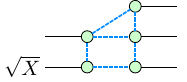}} \hspace{-1mm}\right]$
    \caption{ZX-encoding graphs of subsystem codes that gauge fix into a stabilizer code. From left to right, we have the ZX-encoding graph of a stabilizer code and, in brackets, the subsystem codes that Z/X/Y gauge fix into the stabilizer code.
    }
    \label{fig:subs_XY_gauge_fix}
\end{figure*}
We have chosen to write the states using an X spider and a Clifford gate, rather than using Z spiders. This allows us to remove the $\sqrt{X}$ or $H$ from the gauge wire via local complementations. The remaining X spider removes the gauge vertex and all its edges by the state-copy, colour-change and finally fusion rules of the ZX-calculus. We illustrate the Z gauge fixing below:
\begin{equation}\label{eq:Z_gauge_toy_code}
    \raisebox{-6mm}{\includegraphics{ figures/gauge_fixing/Zgauge.pdf}}
    \hspace{-1.5mm} \hspace{-1mm} 
    \mathrel{\vcenter{\hbox{\includegraphics{ figures/gauge_fixing/Subs_Example.pdf}}}}
    \hspace{1mm} = \hspace{1mm}
    \mathrel{\vcenter{\hbox{\includegraphics{ figures/gauge_fixing/S_one_Example.pdf}}}} \hspace{1mm}.
\end{equation}
Here, this gauge fixing corresponds to measuring the gauge operator $Z_g = IZX$.

Moreover, for every subsystem code $C_{\mathrm{sub}}$ that Z gauge fixes into a given stabilizer code $C_{\mathrm{stab}}$, one can build its ZX-encoding graph $G_{\mathrm{sub}}$ by extending the ZX-encoding graph $G_{\mathrm{stab}}$. 
\begin{definition}\label{def:extending}
    We define extending a ZX-encoding graph $G$ as adding an input gauge qubit that must be connected to at least one physical qubit, and must not have the same connectivity as that of another combination of logical inputs in $G$.
\end{definition}
This definition allows us to prove that any subsystem that gauge fixes into a certain stabilizer code can be found by extending the ZX-encoding graph of the latter.
\begin{theorem}
    \label{thm:construction_subs}
    Let $C_{\mathrm{stab}}$ be a stabilizer code, and $G_{\mathrm{stab}}$ be its ZX-encoding graph. The ZX-encoding graph of any subsystem code $C_{\mathrm{sub}}$ that gauge fixes into $C_{\mathrm{stab}}$ can be built by extending $G_{\mathrm{stab}}$ and adding a local Clifford $I$ (resp. $H$,$\sqrt{X}$) to the gauge input. The resulting ZX-encoding graph $G_{\mathrm{sub}}$ Z (resp. X,Y) gauge fixes into $G_{\mathrm{stab}}$.
\end{theorem}
\begin{proof}
    We begin by proving the statement for any $G_{\mathrm{stab}}$ in normal form into which $G_{\mathrm{sub}}$, also in normal form, Z gauge fixes. Then, the statement for general $G_{\mathrm{stab}}$ \&  $G_{\mathrm{sub}}$ that Z, X or Y gauge fixes into the former follows.

    Any stabilizer operation can be written as a ZX-diagram in the Clifford fragment \cite{backens_zx-calculus_2014}, including all encoding maps. So, for any $C_{\mathrm{stab}}$ and $C_{\mathrm{sub}}$, there exists a ZX-diagram corresponding to their encoding procedure, and it can be transformed into its normal form. 

    By choosing the appropriate qubit ordering of the normal form, we may always assume that a gate on the gauge qubit input in $G_{\mathrm{sub}}$ must be $I$ or $\sqrt{Z}$ \footnote{In the normal form it is impossible for all wires to have an $H$ gate as some will necessarily be cancelled by pivoting.}.
    The $Z$ gauge fixing procedure \raisebox{-0.5mm}{\includegraphics{figures/gauge_fixing/Zgauge.pdf}} uses the state-copy rule of the ZX-calculus, where the Z spider can have any phase possible. Since both $I$ and $\sqrt{Z}$ are Z spiders with a phase in the ZX-calculus, the Z gauge fixing is always possible on $G_{\mathrm{sub}}$, and it will always cut off the gauge qubit from the rest of the graph without any further graph or local Clifford modifications. Because both $G_{\mathrm{stab}}$ and $G_{\mathrm{sub}}$ are in normal form, which are unique, and by construction $C_{\mathrm{sub}}$ Z gauge fixes into $C_{\mathrm{stab}}$, $G_{\mathrm{sub}}$ after the $Z$ gauge fixing must be exactly equal to $G_{\mathrm{stab}}$. 
    We conclude $G_{\mathrm{sub}}$ in normal form is built by \emph{extending} $G_{\mathrm{stab}}$ in normal form, i.e. adding a gauge input.

    We now generalize this to all ZX-encoding graphs $G_{\mathrm{stab}}$ by observing that $G_{\mathrm{stab}}$ in normal form can be transformed into a given $G_{\mathrm{stab}}$ by applying LCs on the input and output vertices. These same LCs can necessarily be applied to $G_{\mathrm{sub}}$ in the normal form without changing the gauge fixing procedure because none of them will be applied to the gauge qubit. Thus, the resulting $G_{\mathrm{sub}}$ will Z gauge fix into $G_{\mathrm{stab}}$ as it did before the LCs.

    For the X and Y gauge fixings, we first recall that $H \cdot H = I = \sqrt{X} \cdot \sqrt{X}$. Then, these Cliffords squared may be spawned on the gauge input of a $G_{\mathrm{sub}}$ that Z gauge fixes into $G_{\mathrm{stab}}$, as in Fig.~\ref{fig:subs_XY_gauge_fix}. 
    As illustrated by the brackets, we consider the inner local Clifford to be part of the ZX-encoding. Notice how the local Clifford leftover on the gauge wire transforms the initial Z gauge fixing into the X and Y gauge fixings of definition~\ref{def:gauge_fixing_zx}. We conclude that the ZX-encoding graphs of the subsystem codes that X or Y gauge fix into $C_{\mathrm{stab}}$ are the extended $G_{\mathrm{stab}}$ with $H$ or $\sqrt{Z}$ respectively added on the gauge input. Considering there are as many subsystems that Z as one that X or Y gauge fix into $C_{\mathrm{stab}}$, we can find the ZX-encoding graph of all subsystems that gauge fix into $C_{\mathrm{stab}}$ using this procedure. 
\end{proof}

\textit{Dynamical automorphisms.---}
As seen in the previous section, gauge fixing can be leveraged to transfer the logical information from one stabilizer code to another.
In fact, logical gates such as lattice surgery and code switching can be understood as gauge fixing~\cite{vuillot_code_2019, aasen_measurement_2023}.
Dynamical automorphism gates~\cite{davydova_quantum_2024} can be understood as a generalisation of code switching where we perform local Clifford gates, qubit permutations and gauge fixings from one stabilizer to another, until we return to the original code, up to a logical gate.
In other words, dynamical automorphisms are produced by nontrivial loops in the space of \([\![n,k]\!]\) stabilizer codes.
Since we know that any two codes can be connected with high probability in a distance-preserving way via code switching \cite[theorem 1]{huang_transversal_2018} (see also \cite[section 2.6]{colladay_rewiring_2018} and \cite{aasen_measurement_2023}),
this space is rich and can be used as a resource.

\captionsetup[subfigure]{labelformat=empty}
\begin{figure*}
    \vspace{-5mm}
    \centering
    \subfloat[\label{subfig:422_zx-enc}]{
        \raisebox{9.5mm}{\llap{{\footnotesize (a)}\hspace{2pt}}}
        $\mathord{\vcenter{\hbox{\includegraphics[height=2cm]{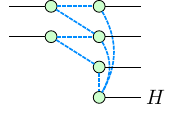}}}}$
    
    } \hspace{0.5cm}
    \subfloat[\label{subfig:422_dyn_auto}]{
        \raisebox{9.5mm}{\llap{{\footnotesize(b)}\hspace{2pt}}}
        $\mathord{\vcenter{\hbox{\includegraphics[height=2cm]{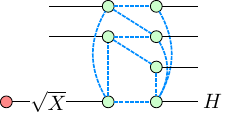}}}}
        \,
        \mathrel{\cong}
        \, \mathord{\vcenter{\hbox{\includegraphics[height=2cm]{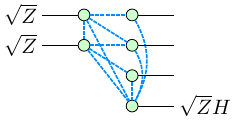}}}}
        \,
        \mathrel{\xrightarrow{\sqrt{X}_1\sqrt{X}_3\sqrt{X}_4}}
        \, \mathord{\vcenter{\hbox{\includegraphics[height=2cm]{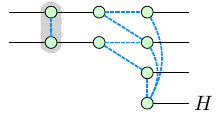}}}}
        $
    }
    \vspace{-7.5mm}
    \caption{Dynamical automorphism of the \nkd{4}{2}{2} code. We extend the original ZX-encoding graph (a) to construct a subsystem code (b) that Y gauge-fixes into the \nkd{4}{2}{2} code up to some local Clifford gates and a logical CZ gate (highlighted in gray).
    }
    \label{fig:422_dyn_auto}
\end{figure*}

In Fig.~\ref{fig:422_dyn_auto} we give an example dynamical automorphism implementing a logical $CZ$ for the \nkd{4}{2}{2} code, which has stabilizers $XXXX$ and $ZZZZ$.
The CZ gate is in fact transversal for this code, so while not necessary here, this example illustrates how gauge fixing induces entangling operations on logical qubits encoded in the same code block.
The CZ gate is implemented through a single gauge fixing where we measure the operator $Y_g = YZYI$, followed by some local Clifford gates which bring the code back to itself.

Performing a logical gate through dynamical automorphisms requires a nontrivial loop in the space of \([\![n,k]\!]\)  stabilizer codes, but in the instance above, a single gauge fixing was sufficient to create such a loop. 
One may wonder where the loop is hiding. 
To seek out the answer, we turn to a simpler example of dynamical automorphism for our 3-qubit code; see Fig.~\ref{fig:3-qubit_example}a. 
In this code, the Y gauge fixing from Fig.~\ref{fig:3-qubit_example}c, combined with some LC gates and a qubit permutation, generates a logical $S$ gate.

\begin{figure}[h]
    \centering
    \subfloat[]{
        \raisebox{4.75mm}{\llap{{\footnotesize (a)}\hspace{2pt}}}
        $\mathord{\vcenter{\hbox{\includegraphics[scale=0.8]{figures/gauge_fixing/S_one_Example.pdf}}}}$
    } 
    \vspace{-5mm}
    \subfloat[]{
        \raisebox{4.5mm}{\llap{{\footnotesize (b)}\hspace{2pt}}}
        \hspace{-5.5mm}
        $\mathord{\vcenter{\hbox{\raisebox{-11mm}{\includegraphics[scale=0.8]{figures/gauge_fixing/Xgauge.pdf}}}}}
        \hspace{-2.5mm}
        \mathord{\vcenter{\hbox{ \includegraphics[scale=0.8]{figures/gauge_fixing/Subs_Example.pdf}} }}
        \cong 
        \mathord{\vcenter{\hbox{ \includegraphics[scale=0.8]{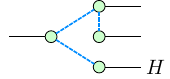}} }}
        \cong
        \mathord{\vcenter{\hbox{ \includegraphics[scale=0.8]{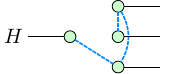}}}}
        $}
    \vspace{-3mm}
    \subfloat[]{
        \raisebox{10.75mm}{\llap{{\footnotesize (c)}\hspace{2pt}}}
        \hspace{-2mm}
        $\begin{gathered}
        \mathord{\vcenter{\hbox{\raisebox{-11mm}{\includegraphics[scale=0.8]{ figures/gauge_fixing/Ygauge.pdf}}}}}
        \hspace{-1mm}
        \mathord{\vcenter{\hbox{ \includegraphics[clip, trim=2mm 0mm 0cm 0mm,scale=0.8]{figures/gauge_fixing/Subs_Example.pdf}}}}
        \cong
        \mathord{\vcenter{\hbox{ \includegraphics[scale=0.8]{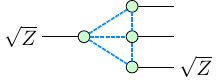}}}}
        \xrightarrow{H_1 H_2 \sqrt{Z}_3} \\
        \mathord{\vcenter{\hbox{ \includegraphics[scale=0.8]{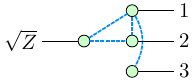}} }}
        \xrightarrow{1\leftrightarrow2}
        \mathord{\vcenter{\hbox{\includegraphics[clip, trim=1mm 0mm 0cm 0mm, scale=0.8]{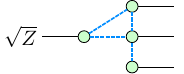}}}}
        \end{gathered}$
    }
    \vspace{-4.5mm}
    \caption{Gauge fixing examples and dynamical automorphism for a 3-qubit code. (a) ZX-encoding graph of this code. (b) X gauge fixing of a subsystem code built by extending (a), giving two equivalent ZX-encoding graphs. (c) Dynamical automorphism generating a logical $S$ gate consisting of a Y gauge fixing, local Clifford gates, and a qubit permutation. }
    \label{fig:3-qubit_example}
\end{figure}
\captionsetup[subfigure]{labelformat=parens}

In both examples (Fig.~\ref{fig:422_dyn_auto}b \& Fig.~\ref{fig:3-qubit_example}c), a final set of LC gates and qubit permutations is required.
Graphs that map to each other through LC gates are called LC equivalent, and the set of all LC equivalent graphs is the LC orbit (we consider permutations to be free, but we note that LC orbits need not be permutation invariant in general).
Because of its small size, it is possible to illustrate the LC orbit of our 3-qubit code; see Fig.~\ref{fig:LC_orbit}a.
\begin{figure}
    \centering
    \includegraphics[clip, trim=0cm 6mm 0cm 6mm, width=1\linewidth]{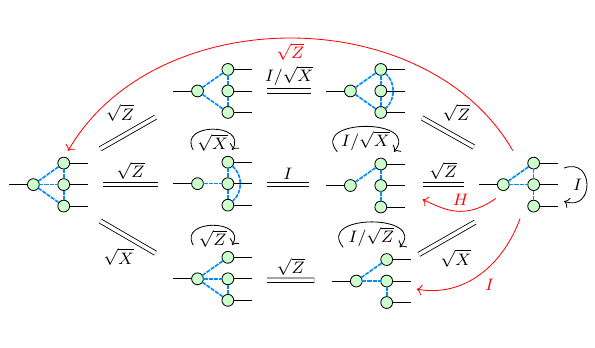}
    \caption{Local Clifford (LC) orbit of the 3-qubit code from Fig.~\ref{fig:3-qubit_example}a, whose ZX-encoding graph is on the far right side of the diagram. Here, the equalities between ZX-encoding graphs refer to LC equivalence, whereas the looping arrows represent LC gates that bring the graph to itself. The gates above them keep track of the logical gate induced by the given LC, and we do not keep track of the physical gates they induce. The red arrows correspond to the X and Y gauge fixings seen in Fig.~\ref{fig:3-qubit_example}b \& Fig.~\ref{fig:3-qubit_example}c, respectively. }
    \label{fig:LC_orbit}
\end{figure}
The logical gates induced by local Clifford gates for a step in the LC orbit are indicated above the edges $\overset{U}{=}$. Similarly, the red arrows {\color{red}$\overset{U}{\rightarrow}$} point to the resulting ZX-encoding graphs after the X and Y gauge fixings, seen in Fig.~\ref{fig:3-qubit_example}b and Fig.~\ref{fig:3-qubit_example}c respectively. As such, this diagram reveals how gauge fixing gives shortcuts, or chutes and ladders, in the LC orbit graph of a code.

To show this, we start by looking at the static QEC picture. We compute logical gates in the automorphism group of this code by composing the LC induced gates along a path starting and ending at a given graph. All possible paths for the ZX-encoding graph on the far right of Fig.~\ref{fig:LC_orbit} generate the group \{$I,H$\}. Since all graphs are LC equivalent, any other choice of ZX-encoding graph would give the same set of gates in another basis. Thus, the number of generators in an automorphism group and their periods are invariant under local complementation. 

We move on to dynamical automorphisms by considering the gauge fixings. Starting from the same graph as above, we compute the paths that use the gauge fixings described earlier by composing each step's induced gate. Both gauge fixings combined produce \{$\sqrt{Z}, \sqrt{X}\sqrt{Z}$\}. These expand the automorphism group such that together they generate all local Clifford gates, 
showing that the ZX-calculus gauge-fixing formalism 
can generate additional gates beyond those in the code's automorphism group. 

Gauge fixing can also project the logical information of one code into a \textit{different} code, which may enlarge the possible dynamical automorphisms to, in principle, any logical Clifford operators. 
In the End Matter, we build a distance-preserving dynamical automorphism implementing an \(S\) gate for the 7-qubit bare code, whose automorphism group does not induce any logical action, using the map of the space of \nkd{7}{1}{3} stabilizer codes \cite{cross_small_2025}.

Finally, one might wonder how entangling gates compare with gauge fixing. In Section \ref{sec:clifford_conversion} of the SM, we investigate the difference between gauge fixing and Clifford conversion~\cite{hill_fault-tolerant_2013,hwang_fault-tolerant_2015}, code switching using 2-qubit Clifford gates. 
We prove that every Clifford conversion can be replaced by a single gauge-fixing step, and give an example where a single 2-qubit Clifford cannot replace gauge fixing.

\textit{Conclusion.---}
We have introduced a formalism for gauge fixing within the ZX-calculus, which enables us to construct dynamical automorphisms generating logical Clifford gates for any stabilizer code. 
Using this formalism, we implement logical CZ for the \nkd{4}{2}{2} code, and a logical phase gate for a 3-qubit code and the \nkd{7}{1}{3} bare code. 

An important extension of our work is to compute the effect of the gauge fixing steps on the code distance, which may decrease in some cases~\cite{vuillot_code_2019}.
To accurately benchmark the performance of our dynamical automorphism gates, we would also need to construct circuits for measuring the gauge operators and include rounds of stabilizer measurement. 

Our formalism faciliates the search for specific dynamical autormorphisms.
Once the space of local Clifford equivalences and gauge fixings are mapped out, then we can perform a double-ended breadth-first search from the initial ZX-encoding graph and the final ZX-encoding graph, which we construct by appending the ZX-diagram of the desired Clifford gate to the initial graph.
A particularly interesting case would be to consider codes with transversal $T$ gates and then complete a universal gate set by implementing Clifford gates using dynamical automorphisms or to develop new measurement-based implementations of non-Clifford gates \cite{davydova_quantum_2024}.

\textit{Acknowledgements.---}
The authors thank Aleksander Kubica for pointing out the work by Colladay and Mueller \cite{colladay_rewiring_2018}, and to Amolak Ratan Kalra for pointing out the ZX-calculus perspective and associated references \cite{backens_zx-calculus_2014} and \cite{hu_improved_2022,khesin_equivalence_2024,khesin_equivalence_2024}.

Research at Perimeter Institute is supported in part by the Government of Canada through the Department of Innovation, Science and Economic Development Canada and by the Province of Ontario through the Ministry of Colleges and Universities.
MV was supported in part by Plan
France 2030 through the project ANR-22-PETQ0006.
SZB was supported by a NSERC Canada Graduate Scholarship (Master's), and a Master's scholarship from the Fonds de recherche du Québec -- Nature and Technologies (FRQNT).
VVA acknowledges NSF grants OMA2120757 (QLCI) and CCF2104489.
Certain products, commercial and otherwise, are mentioned in this publication. These mentions are for informational purposes only, and do not imply recommendation or endorsement by NIST.

\bibliography{references}

\section*{End Matter}\label{sec:end_matter}

Here, we provide an example of a dynamical automorphism generating a logical $S$ gate for the \nkd{7}{1}{3} bare code \cite{yu_graphical_2007,muyuan2017bare}. 
This dynamical automorphism differentiates itself from the ones presented previously because it changes code space. In addition, each code in the sequence has distance 3, bringing us one step closer to fault tolerance. 
The full dynamical automorphism is shown in Fig.~\ref{fig:bare_code_ex}. 

The procedure was constructed using Cross and Vandeth's map of \nkd{7}{1}{3} codes, where they identify all subsystem codes which map one \nkd{7}{1}{3} code to another \cite{cross_small_2025}. From their data about the stabilizers and logical operators of each stabilizer or subsystem code, we were able to construct the corresponding ZX-encoding graphs by hand. The method used is outlined in Subsection \ref{subsec:const_zx-enc} of the SM. The last step was to compute the dynamical autormophism using our formalism. To do this, we adapted a Python package Graph state compass \cite{morley-short_gsc_2019, adcock_mapping_2020} that computes LC orbits of graphs states, in order to find the LC orbits of the ZX-encoding graphs and the LC gates and qubit permutations that map one graph to another. This allowed us to piece together the dynamical automorphism, and push back all the intermediate LC gates to the end of the procedure by using the appropriate representative in the LC orbit.

\begin{figure*}[t]
    \centering
    \subfloat[]
        {
            \includegraphics[scale=0.7]{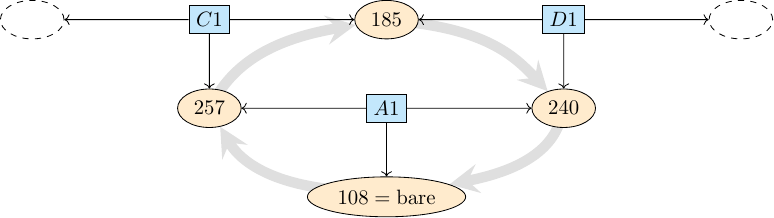}
        }
        \hspace{1cm}
    \subfloat[]
        {
            \hspace{1cm}\includegraphics[scale=0.7]{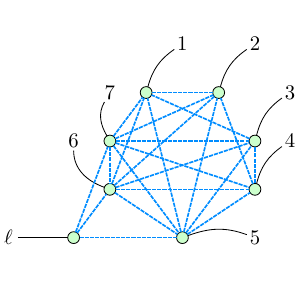}\hspace*{1cm}
        }
    \linebreak
    \subfloat[]{
        $
        \mathord{\vcenter{\hbox{ \includegraphics[scale=0.7]{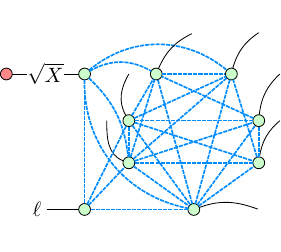}} }}
        \hspace{5mm} = \hspace{-5mm}
        \mathord{\vcenter{\hbox{ \includegraphics[scale=0.7]{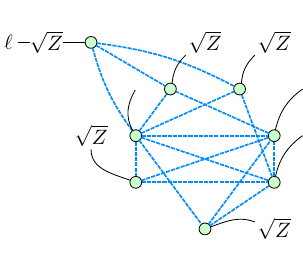}} }}
        $
    }\hfill
    \subfloat[]{
        $
        \mathord{\vcenter{\hbox{ \includegraphics[scale=0.7]{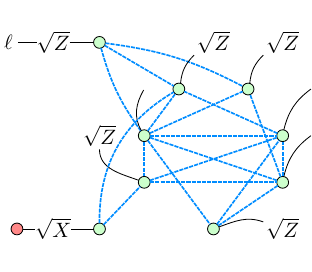}} }}
        \hspace{5mm} = \hspace{-5mm} 
        \mathord{\vcenter{\hbox{ \includegraphics[scale=0.7]{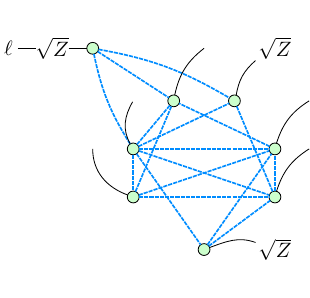}} }}
        $
    }
    \linebreak
    \subfloat[]{
        $
        \mathord{\vcenter{\hbox{ \includegraphics[scale=0.7]{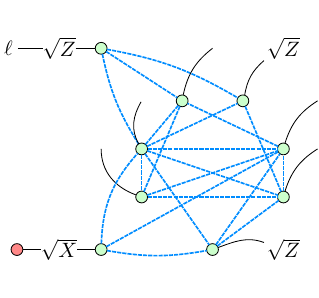}} }}
        \hspace{5mm} = \hspace{-5mm} 
        \mathord{\vcenter{\hbox{ \includegraphics[scale=0.7]{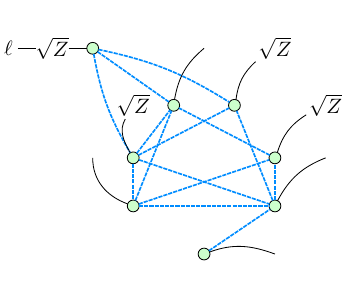}} }}
        $
    }\hfill
    \subfloat[]{
        $
        \mathord{\vcenter{\hbox{ \includegraphics[scale=0.7]{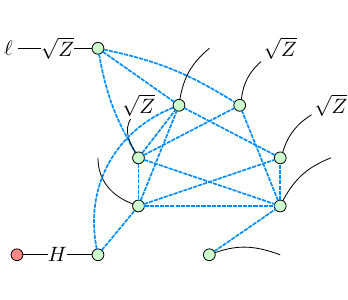}} }}
        \hspace{5mm} = \hspace{-5mm} 
        \mathord{\vcenter{\hbox{ \includegraphics[scale=0.7]{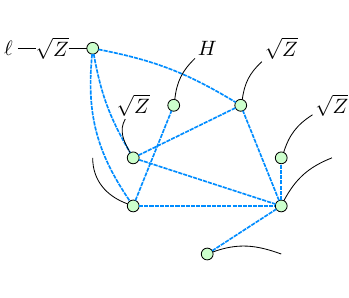}} }}
        $
    }
    \linebreak
    \subfloat[]{
        $
        \mathord{\vcenter{\hbox{ \includegraphics[scale=0.7]{figures/bare_code/After_Gauge_bare.pdf}}}}
        \xrightarrow{\binom{\sqrt{Z}_1 H_1 \sqrt{Z}_2 \sqrt{Z}_3}{\sqrt{X}_4\sqrt{Z}_4 \sqrt{X}_6 \sqrt{Z}_7} }
        \mathord{\vcenter{\hbox{ \includegraphics[scale=0.7]{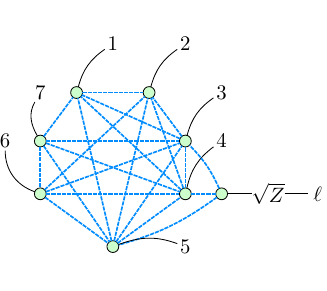}} }}
        \xrightarrow{(164)(2753)}
        \mathord{\vcenter{\hbox{ \includegraphics[scale=0.7]{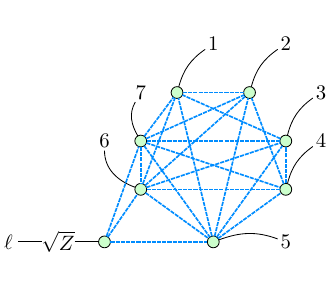}}}}
        $
    }
    \caption{
        Dynamical automorphism for the bare code implementing a logical $S$ gate:
        (a) Closed loop in the space of \nkd{7}{1}{3} stabilizer codes.
        The blue squares are subsystem codes that $Z/X/Y$ gauge fix into stabilizer codes (orange bubbles).
        See {\cite[Figure~12]{cross_small_2025}} for more details.
        (b) ZX-encoding graph of the \nkd{7}{1}{3} bare code \cite{yu_graphical_2007}.
        We label the logical qubit wire $\ell$ and the physical qubit wires $1,\ldots,7$.
        (c) $Y$-gauge fixing of subsystem $A1$ from the bare code $(108)$ into stabilizer code $(257)$.
        (d) $Y$-gauge fixing of $C1$ from $(257)$ into $(185)$.
        (e) $Y$-gauge fixing of $D1$ from $(185)$ into $(240)$.
        (f) $X$-gauge fixing of $A1$ from $(240)$ back into the bare code (up to local Cliffords and qubit permutation).
        (g) Final local Clifford gates and qubit permutations into the bare code $(108)$.
    }
    \label{fig:bare_code_ex}
\end{figure*}

\beginsupplementary
\section*{Supplementary material}

In the Supplementary Material (SM), we begin with Section \ref{sec:notation_conventions} by defining the symplectic representation, a convention used throughout the Letter and the SM. Then, Section \ref{sec:zx-calculus} gives an introduction to the ZX-calculus. This leads the reader to Section \ref{sec:subs_zx}, where we discuss in more detail the representation of stabilizer and subsystem codes in the ZX-calculus, i.e. ZX-encoding graphs. More precisely, we start by defining ZX-encoding graphs for subsystem codes (\ref{subsec:zx-enc_subs}). We then explain why their biadjacency matrix must be full rank (\ref{subsec:zx-enc_full_rank}). We outline how to get the normal form of a ZX-encoding graph (\ref{subsec:zx-enc_norm_form}), and compare this normal form with other equivalent representations (\ref{subsec:gslc_ap_norm_form}). Finally, we discuss the construction of ZX-encoding graphs from the code's stabilizers and logical operators (\ref{subsec:const_zx-enc}). In Section \ref{sec:clifford_conversion}, the last of the SM, we compare gauge fixing to Clifford conversions. After defining these types of transformations, we prove that every code conversion can be replaced by a gauge fixing (\ref{subsec:cliff_to_gauge}). We then explore how $\pi/4$ rotations on the gauge degree of freedom, equivalent to gauge fixing, may be highly entangling (\ref{subsec:gauge_entang}). We conclude the SM by giving a simple argument justifying that not every gauge fixing may be replaced by a Clifford conversion (\ref{subsec:code_deform}).

\section{Symplectic representation}
\label{sec:notation_conventions}

We let $\paulis(n)$ denote the Pauli group (with phases) on $n$ physical qubits,
and specify subsystem codes (including stabilizer codes) via stabilizer tableaus \cite{aaronson_improved_2004}
\[
    \renewcommand{\arraystretch}{1.2}
    \left[\begin{array}{c|ccc||ccc}
        \multirow{2}{*}{$s_1~~\ldots~s_m$}
        & r_1 & \ldots & r_g
        & \logicalX_1 & \ldots & \logicalX_k
        \\
        \cline{2-7}
        & t_1 & \ldots & t_g
        & \logicalZ_1 & \ldots & \logicalZ_k
    \end{array}\right]\subset\paulis(n)
\]
where $r_i,s_i,t_i,\logicalX_i,\logicalZ_i \in\paulis(n)$ denote (the generators of) the stabilizer, gauge checks, and logical Paulis of the subsystem code, respectively.
Paulis thereby appearing in the same column of the stabilizer tableau anticommute, and commute otherwise.
Further we denote the group generated by some collection of Paulis by
\[
    \langle g_1,\ldots g_m\rangle \subset \paulis(n)
\]
and the center of a group by $Z(G)\leq G$.
With this in mind, we abbreviate the group of stabilizers by $\centerstab=\langle s_1,\ldots,s_m\rangle$ (as the center of the subsystem code),
and with mild abuse of notation we will often write
\[
    \renewcommand{\arraystretch}{1.2}
    \left[\begin{array}{c|ccc||ccc}
        \multirow{2}{*}{$\centerstab$}
        & r_1 & \ldots & r_g
        & \logicalX_1 & \ldots & \logicalX_k
        \\
        \cline{2-7}
        & t_1 & \ldots & t_g
        & \logicalZ_1 & \ldots & \logicalZ_k
    \end{array}\right]\subset\paulis(n).
\]
We denote the Clifford group by $\cliffords(n)$
and the group of local Cliffords by $\LC:=\cliffords(1)\otimes\ldots\otimes\cliffords(1)\leq\cliffords(n)$\footnote{With mild abuse of notation, we will occasionally also identify single qubit Cliffords as local Cliffords, $\LC=\cliffords(1)$.}.
Note that the action of Cliffords via conjugation
\[
    \ad(U):\paulis(n)\to\paulis(n):\tabspace \ad(U)P:=UPU^\dagger
\]
defines a $*$-automorphism
and with the subgroup of Paulis freely modifying signs:
\[
    (U=X:~XZX = -Z),
    \tabspace
    (U=Z:~ZXZ = -X).
\]
The Pauli group defines however, as a subgroup $\paulis(n)\subset\LC$,
a free resource from the fault tolerance perspective.
Thus, using the symplectic representation,
\[
    \mathrm{F}_2^n\oplus \mathrm{F}_2^n:
    \tabspace
    X_i = e_i\oplus 0,
    \tabspace
    Z_j = 0\oplus e_j,
\]
we may freely identify the Clifford group (mod Paulis) with the symplectic group,
\[
    \ad\cliffords(n)/\ad\paulis(n) = \symplecticgroup(2n),
\]
where the base field, $\mathrm{F}_2$, is suppressed in the notation.
With this, phase gates square to the identity in the symplectic representation
\[
    \ad(S)^2=\ad(S^2)=\ad(Z)=1\tabspace \mathrm{mod}~\ad\paulis(n)
\]
and as such we do not need to worry about their adjoints.
We therefore denote the Hadamard and phase gate in the symplectic representation by
\[
    H=\begin{pmatrix}
        0 & 1 \\
        1 & 0
    \end{pmatrix},
    \tabspace
    \sqrt{Z}=\begin{pmatrix}
        1 & 1 \\
        0 & 1
    \end{pmatrix},
    \tabspace
    \sqrt{X}:= H \sqrt{Z} H=
    \begin{pmatrix}
        1 & 0 \\
        1 & 1
    \end{pmatrix}
\]
to visually differentiate these symplectic phase gates from their usual counterparts.
These satisfy the following relations in the symplectic representation:
\begin{equation}
    \label{eq:sympl_repr_relations}
    H = \sqrt{X}\sqrt{Z}\sqrt{X} = \sqrt{Z}\sqrt{X}\sqrt{Z}\tabbspace \left(H^2=\sqrt{X}^2=\sqrt{Z}^2=1\right).
\end{equation}
In particular, these are all order-2 when considered in the symplectic representation.\\
All of the above carries over to the logical level,
since one may always realize every logical Paulis via physical Paulis,
such as already in the subsystem code above
\[
    \logicalX_1, \ldots,\logicalX_k,~
    \logicalZ_1, \ldots,\logicalZ_k~
    \in\paulis(n).
\]
As such we may freely work in the symplectic representation
at both the physical and logical level.
By mild abuse of notation we will thus freely identify throughout this work
\[
    \cliffords(n)\cong\symplecticgroup(2n)
    \tabspace
    \mathand
    \tabspace
    \cliffords(k)\cong\symplecticgroup(2k).
\]
Further note that one may complete any stabilizer tableau by some set of destabilizers~\cite{aaronson_improved_2004}
\[
    \renewcommand{\arraystretch}{1.2}
    \left[\begin{array}{ccc}
        r_1 & \ldots & r_n \\
        \hline
        d_1 & \ldots & d_n
    \end{array}\right]
    :=
    \left[\begin{array}{ccc|ccc||ccc}
        r_1 & \ldots & r_m &
        s_1 & \ldots & s_g &
        \logicalX_1 & \ldots & \logicalX_k
        \\
        \hline
        d_1 & \ldots & d_m &
        t_1 & \ldots & t_g &
        \logicalZ_1 & \ldots & \logicalZ_k
    \end{array}\right]
\]
and as such we may also freely alter any sign via conjugation by Paulis $\paulis(n)\subset\cliffords(n)$:
\[
    \renewcommand{\arraystretch}{1.2}
    s_i\left[\begin{array}{ccc}
        \ldots~s_i~\ldots \\
        \hline
        \ldots~d_i~\ldots
    \end{array}\right]s_i^\dagger
    =
    \left[\begin{array}{ccc}
        \ldots~ s_i ~\ldots \\
        \hline
        \ldots~ -d_i ~\ldots
    \end{array}\right]
    \tabspace
    \mathand
    \tabspace
    d_i\left[\begin{array}{c}
        \ldots~s_i~\ldots \\
        \hline
        \ldots~d_i~\ldots
    \end{array}\right]d_i^\dagger
    =
    \left[\begin{array}{c}
        \ldots~ -s_i ~\ldots \\
        \hline
        \ldots~ d_i ~\ldots
    \end{array}\right].
\]
Thus, we also do not need worry about signs (or complex phases) in subsystem codes.
Finally, we denote the measurement of Paulis, meant as gauge checks, by $M(t)=\frac{1\pm t}{2}$.

\section{Primer on the ZX-Calculus}\label{sec:zx-calculus}

The ZX-calculus, originally introduced by Coecke and Duncan \cite{coecke_interacting_2011}, is a tool used to describe quantum circuits through a graphical representation. We refer the reader to \cite{kissinger_picturing_2024} for a detailed exposition. In the ZX-calculus, vertices and edges are called spiders and wires. There are two different colours of spiders: green Z spiders and red X spiders. They can have as many input and output wires as one might want:
\begin{gather}
    \mathord{\vcenter{\hbox{\includegraphics{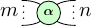}}}}  = \ket{0}^{\otimes n} \bra{0}^{\otimes m} + e^{i\alpha} \ket{1}^{\otimes n} \bra{1}^{\otimes m},\\
    \mathord{\vcenter{\hbox{\includegraphics{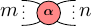}}}} = \ket{+}^{\otimes n} \bra{+}^{\otimes m} + e^{i\alpha} \ket{-}^{\otimes n} \bra{-}^{\otimes m}.
\end{gather}
Notice how the number of input and output wires in a spider dictates the number of tensor products that form the input and output states of a linear map. In other words, spiders are tensors with varying dimensions, and connecting spiders together through wires corresponds to tensor contractions. So ZX-diagrams are tensor networks.

More concretely, if we have $n=1$ and $m=0$, then these spiders correspond to the superposition of states in the computational or the Hadamard basis with a relative phase $\alpha$. For example, the $\ket{0}$ state can be represented as \includegraphics{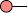}.
These one-legged spiders are also used to represent projective measurements: 
\begin{equation}
    \mathord{\vcenter{\hbox{\includegraphics{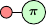}}}} 
    = \braket{-|0} = \frac{1}{\sqrt{2}}.
\end{equation}
If instead we have a spider with $n=m=1$ and $\alpha=\pi$, then they represent the familiar $X$ and $Z$ Pauli matrices:
\begin{equation}
    \mathord{\vcenter{\hbox{\includegraphics{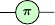}}}} = \ket{0}\bra{0} - \ket{1}\bra{1} = Z 
    \hspace{3mm}   \& \hspace{3mm}  
    \mathord{\vcenter{\hbox{\includegraphics{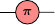}}}} = \ket{+}\bra{+} - \ket{-}\bra{-} = X. 
\end{equation}
Spiders with different phases give a description of other known gates, notably
\begin{equation}
    \mathord{\vcenter{\hbox{\includegraphics{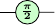}}}} = \ket{0}\bra{0} +i \ket{1}\bra{1} = S = \sqrt{Z}
    \hspace{3mm}   \& \hspace{3mm}
    \mathord{\vcenter{\hbox{\includegraphics{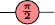}}}} = \ket{+}\bra{+} +i \ket{-}\bra{-} = \sqrt{X}.
\end{equation}
Composing these gates, we can build a Hadamard gate:
\begin{equation}
    H = \sqrt{Z}\sqrt{X}\sqrt{Z} = \mathord{\vcenter{\hbox{\includegraphics{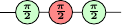}}}} 
    \hspace{3mm} \text{or} \hspace{3mm}
     H = \sqrt{X}\sqrt{Z}\sqrt{X} = \mathord{\vcenter{\hbox{\includegraphics{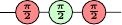}}}} .
\end{equation}
This gate deserves its own symbol as it is frequently used: 
$ H = \mathord{\vcenter{\hbox{\includegraphics{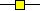}}}} $.
Using the phase gate and the Hadamard gate, one can now build all 1-qubit Clifford circuits using spiders with phases $\alpha \in \frac{\pi}{2} \mathbb{Z}$. Circuits made up of only such spiders are in the Clifford fragment. We work within this set of ZX-diagrams for the entirety of the manuscript. 

To move on to circuits supported on more qubits, we introduce the rules of the ZX-calculus in Figure \ref{fig:ZX-rules}.
\begin{figure}[t]
    \centering
        \subfloat[]{
            \label{subfig:rules-spider-fusion}
            \centering
            \begin{minipage}[t][2cm][c]{0.275\linewidth}
            \[
                \mathrel{\vcenter{\hbox{ \includegraphics{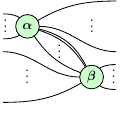}}}}  
                \overset{({\color{blue}f})}{=} \hspace{1mm}
                \mathrel{\vcenter{\hbox{ \includegraphics{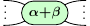}}}}  
            \]
            \end{minipage}
        }     
        \subfloat[]{
            \label{subfig:rules-color-change}
            \centering
            \begin{minipage}[t][2cm][c]{0.275\linewidth}
            \[
                \mathrel{\vcenter{\hbox{ \includegraphics{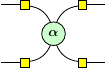}}}}  
                \hspace{2mm} \overset{({\color{blue}cc})}{=} \hspace{1mm}
                \mathrel{\vcenter{\hbox{ \includegraphics{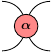}}}}  
            \]
            \end{minipage}
        } 
        \subfloat[]{
            \label{subfig:rules-identity}
            \centering
            \begin{minipage}[t][2cm][c]{0.275\linewidth}
                \[
                \mathrel{\vcenter{\hbox{ \includegraphics{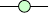}}}}   
                \hspace{1mm} \overset{({\color{blue}i1})}{=}
                \mathrel{\vcenter{\hbox{ \includegraphics{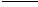}}}}   
                \hspace{1mm} \overset{({\color{blue}i2})}{=} 
                \mathrel{\vcenter{\hbox{ \includegraphics{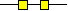}}}}  
                \]
            \end{minipage}
        }
        \vspace{1mm}
        \subfloat[]{
            \label{subfig:rules-pi-copy}
            \centering
            \begin{minipage}[t][2cm][c]{0.275\linewidth}
            \[
            \mathrel{\vcenter{\hbox{ \includegraphics{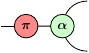}}}}  
            \hspace{2mm} \overset{({\color{blue}\pi})}{=} 
            \mathrel{\vcenter{\hbox{ \includegraphics{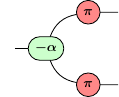}}}}  
            \]
            \end{minipage}
        }
        \subfloat[]{
            \label{subfig:rules-state-copy}
            \centering
            \begin{minipage}[t][2cm][c]{0.275\linewidth}
            \[
            \mathrel{\vcenter{\hbox{ \includegraphics{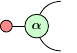}}}}  
            \overset{({\color{blue}sc})}{=} \hspace{2mm} 
            \mathrel{\vcenter{\hbox{ \includegraphics{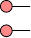}}}}  
            \]
            \end{minipage}
        }
         \subfloat[]{
            \label{subfig:rules-bialgebra}
            \centering
            \begin{minipage}[t][2cm][c]{0.275\linewidth}
            \[
            \mathrel{\vcenter{\hbox{ \includegraphics{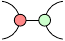}}}}   
            \hspace{2mm} \overset{({\color{blue}b})}{=} \hspace{1mm}
            \mathrel{\vcenter{\hbox{ \includegraphics{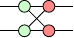}}}}  
            \]
            \end{minipage}
        }
        \caption{Rules of the ZX-calculus. (a) Spider fusion. (b) Color change. (c) Identity. (d) $\pi$-commutation. (e) State-copy. (f) Bialgebra.}
        \label{fig:ZX-rules}
\end{figure}
Most of the ZX-calculus rules translate the familiar quantum mechanics into graphs: the spider fusion rule found in Figure (\ref{subfig:rules-spider-fusion}) is equivalent to $Z = \sqrt{Z} \cdot \sqrt{Z}$; 
the colour change rule 
(\ref{subfig:rules-color-change})
generalizes to tensor networks the action of the Hadamard gate on the Paulis, $HX = ZH $; 
the identity rules (\ref{subfig:rules-identity}{\color{blue}1}) \& (\ref{subfig:rules-identity}{\color{blue}2}) trivially come from $I = \ket{0}\bra{0}+\ket{1}\bra{1}$ and $H^2=I$; 
the $\pi$-commutation ($\ref{subfig:rules-pi-copy}$) rule reflects the Pauli property $XZ=-ZX$; 
and the state-copy rule (\ref{subfig:rules-state-copy}) is associated with the action of measurements of a Bell pair, where two entangled qubits become uncorrelated after a measurement in the computational basis, for example. 

These rules are sufficient to manipulate any ZX-diagram. For example, we introduce the CNOT circuit. One can recognize the propagation of $Z$ and $X$ errors acting on the control qubit:
\begin{equation}
    \mathrel{\vcenter{\hbox{ \includegraphics{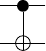}}}} 
    \hspace{1mm} =  
    \mathrel{\vcenter{\hbox{ \includegraphics{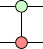}}}} 
    \hspace{3mm} : \hspace{2mm}
    \mathrel{\vcenter{\hbox{ \includegraphics{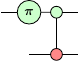}}}}
    \hspace{1mm} \overset{(\ref{subfig:rules-spider-fusion})}{=}
    \mathrel{\vcenter{\hbox{ \includegraphics{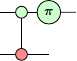}}}}
    \hspace{3mm} \& \hspace{3mm}
    \mathrel{\vcenter{\hbox{ \includegraphics{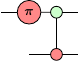}}}} 
    \hspace{1mm} \overset{(\ref{subfig:rules-pi-copy})}{=}
    \mathrel{\vcenter{\hbox{ \includegraphics{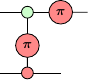}}}} 
    \hspace{1mm} \overset{(\ref{subfig:rules-spider-fusion})}{=}
    \mathrel{\vcenter{\hbox{ \includegraphics{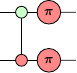}}}} 
    \hspace{2mm}.
\end{equation}
From this ZX-diagram, it is straightforward to find the one for CZs:
\begin{equation}
    \mathrel{\vcenter{\hbox{ \includegraphics{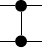}}}} 
    \hspace{1mm} =  
    \mathrel{\vcenter{\hbox{ \includegraphics{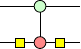}}}} 
    \hspace{1mm} \overset{(\ref{subfig:rules-color-change})}{=}
    \mathrel{\vcenter{\hbox{ \includegraphics{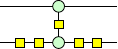}}}}
    \hspace{1mm} \overset{(\ref{subfig:rules-identity}{\color{blue}2})}{=}
    \mathrel{\vcenter{\hbox{ \includegraphics{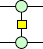}}}}
    \hspace{1mm} =
    \mathrel{\vcenter{\hbox{ \includegraphics{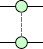}}}} 
    \hspace{2mm}.
\end{equation}
In the last equality above, we have introduced a new notation, Hadamard edges \cite{duncan_graph-theoretic_2020}
$\mathrel{\vcenter{\hbox{ \includegraphics{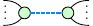}}}} \hspace{1mm} := \mathrel{\vcenter{\hbox{ \includegraphics{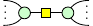}}}} \hspace{1mm}$.
Unlike normal black wires, the spider fusion property does not work on these wires due to the implicit Hadamard gate. This notation allows for the reduction of ZX-diagram to a graph-like diagram. Such diagrams are made up of Z spiders and Hadamard edges that keep the Z spiders from fusing, thus retaining a graph-like shape. 
\begin{definition}[\cite{duncan_graph-theoretic_2020}]
    A ZX-diagram is graph-like when:
    \begin{enumerate}
        \item All spiders are Z-spiders.
        \item Z-spiders are only connected via Hadamard edges.
        \item There are no parallel Hadamard edges or self-loops.
        \item Every input or output edge is connected to a Z-spider and every Z-spider is connected to at most one input or output edge.
    \end{enumerate}
\end{definition}

Graph-like ZX-diagrams come with some extra rules. The action of local Cliffords on an open wire is equivalent to local complementations. 
\begin{definition}[Local complementation \cite{backens_zx-calculus_2014}]
    \label{def:local-complementation}
    Let G = $(V,E)$ be a graph. The local complementation about the vertex $v$ is the operation that inverts the subgraph generated by the neighbourhood of $v$ (but not including $v$ itself). Formally, a local complementation about $v \in V$ sends $G$ to the graph 
    \begin{equation}
        G \star v = (V,E \bigtriangleup \{\{b, c\} | \{b, v\}, \{c, v\} \in E \wedge b \neq c \}),
    \end{equation}
    where $\bigtriangleup$ denotes the symmetric set difference, i.e. $A \bigtriangleup B$ contains all elements that are contained either in $A$ or in $B$ but not in both.
\end{definition}
Concretely, given a graph-like ZX-diagram, a local complementation about a vertex $v$ is equivalent to applying a $\sqrt{X}$ on $v$ and $\sqrt{Z}$ on all its neighbours. This was first found in the context of graph states by Van den Nest et al. \cite{nest_graphical_2004}, and was introduced to the ZX-calculus by Duncan and Perdrix \cite{duncan_graph_2009}.
\begin{theorem}[Van den Nest \cite{nest_graphical_2004}]\label{thm:LC_equiv}
    Given a graph $G$ and a vertex $v \in V(G)$ with a set of neighbouring vertices $N$, 
    \[
        R_x^{(v)}(-\pi/2) \underset{i\in N}{\bigotimes} R_z^{(i)}(\pi/2) \ket{G} = \ket{G \star v}.
    \]
\end{theorem}
As a reminder, for the purposes of this manuscript, we work in the symplectic representation (see section \ref{sec:notation_conventions}). So, we do not distinguish between $R_x^{(v)}(\pm\pi/2)$, i.e between $\sqrt{X}$ and $(\sqrt{X})^{-1}$. This is because we will not care to distinguish between the $\pm1$ phase difference generated by applying these gates on the ZX-encoding graphs.

We illustrate this rule with the following example:
\begin{equation}\label{eq:rules-LC}
    \text{Given} \hspace{2mm} 
    G_{ZX}= \hspace{2mm}
    \mathrel{\vcenter{\hbox{ \includegraphics{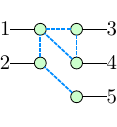}}}} 
    \hspace{2mm}, \hspace{2mm}
    \mathrel{\vcenter{\hbox{ \includegraphics{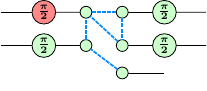}}}} 
    \hspace{1mm} =
    \mathrel{\vcenter{\hbox{ \includegraphics{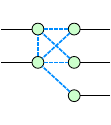}}}} 
    \hspace{2mm}.
\end{equation}
On the left-hand side of the equation, we applied 
$\sqrt{X} = \mathrel{\vcenter{\hbox{ \includegraphics{figures/zx-primer/sqrtX_ZX.pdf}}}} $ 
on vertex 1, and 
$\sqrt{Z} = \mathrel{\vcenter{\hbox{ \includegraphics{figures/zx-primer/sqrtZ_ZX.pdf}}}}$
on all its neighbours, vertices 2, 3 and 4. On the right-hand side, we applied a local complementation about vertex 1. We illustrate the action on the subgraph composed of vertex 1's neighbours:
\begin{equation}
    \mathrel{\vcenter{\hbox{ \includegraphics{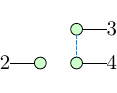}}}}
    \hspace{2mm} \mapsto \hspace{2mm}
    \mathrel{\vcenter{\hbox{ \includegraphics{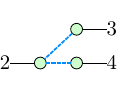}}}}
    \hspace{2mm}.
\end{equation}
The local complementation toggles the edges of this subgraph: the edge between vertices 3 and 4 in $G^{ZX}$ was removed, but the edges between vertices 2 and the two other vertices were added because they weren't connected in $G^{ZX}$. 

The local complementation rule will often be used in a context where one is applying a $\sqrt{X}$ gate on a given vertex, without the $\sqrt{Z}$s. Then, we may adjust example \eqref{eq:rules-LC} by moving these gates to the left-hand side of the equation instead, using the convention from the symplectic representation $Z =(\sqrt{Z})^2 = \sqrt{Z}(\sqrt{Z})^{-1} =I$.

Another useful rule is the equivalence between Hadamard gates on two vertices and pivoting around their common edge \cite{duncan_pivoting_2014}. Pivoting, or edge complementation, is a composition of three local complementations \cite{bouchet_graphic_1988} about two vertices $x, y \in V$, 
\begin{equation}
    G \land xy =  G \star x \star y \star x = G \star y \star x \star y.
\end{equation}
However this series of operations can be performed in one go \cite{bouchet_graphic_1988}:
\begin{equation}\label{eq:EC}
    \text{Given} \hspace{2mm} G~=~
    \mathrel{\vcenter{\hbox{ \includegraphics{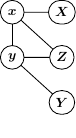}}}}
    ~, \hspace{2mm}
    G \land xy ~=~
    \mathrel{\vcenter{\hbox{ \includegraphics{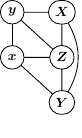}}}}
     ~,
     \vspace{-2mm}
\end{equation}
where $X$ is the set of vertices connected only to $x$, $Y$ is the set of vertices connected only to $y$, and $Z$ is the set of vertices connected to both $x$ and $y$. On the right-hand side of the equation, we first point out that pivoting swaps the two vertices $x$ and $y$. Also, we have toggled the edges between all three sets, but not those within a set. In general, there may be vertices connected to neither $x$ nor $y$. Their edges remain untouched by edge complementation.

We illustrate an example of the equivalence between $H_x \otimes H_y$ and $G_{ZX} \land xy$ in ZX-diagrams, keeping in mind we do not keep track of the Pauli operators in our conventions.
\begin{equation}\label{eq:rules-EC}
    \text{Given} \hspace{2mm} G_{ZX}= \hspace{2mm}
    \mathrel{\vcenter{\hbox{ \includegraphics{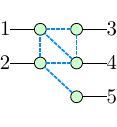}}}}
    \hspace{2mm}, \hspace{2mm}
    \mathrel{\vcenter{\hbox{ \includegraphics{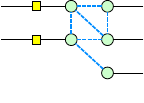}}}}
    \hspace{1mm} =
    \mathrel{\vcenter{\hbox{ \includegraphics{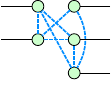}}}}
    \hspace{2mm}.
\end{equation}
Here, because the order of vertices is used as a label, we do not permute vertices 1 \& 2 like in equation \eqref{eq:EC}, but rather we exchange their connectivity.

\section{Stabilizer and subsystem codes via the ZX-calculus}\label{sec:subs_zx}

\subsection{Subsystem codes: ZX-encoding graphs}\label{subsec:zx-enc_subs}

In this subsection, we introduce subsystem codes from the perspective of \zxcalculus.
For this we begin with the observation that any subsystem code
equivalently specifies a Clifford encoding isometry
\begin{equation}
    U:\bigl(\complex^2\bigr)^{\otimes k}\! \otimes\! \bigl(\complex^2\bigr)^{\otimes g} \longrightarrow \bigl(\complex^2\bigr)^{\otimes n}:
    \tabbspace
    U^\dagger U= I\otimes I,
\end{equation}
satisfying the relations
\begin{equation}
    \label{eq:cliff_encoding_isometry}
    s_iU=U(I\otimes I)
    \tabbspace
    \begin{aligned}
        X_{\ell_i} U &= U(X_i\otimes I)\\
        Z_{\ell_i} U &= U(Z_i\otimes I)
    \end{aligned}
    \tabbspace
    \begin{aligned}
        r_i U &= U(I\otimes X_i) \\
        t_i U &= U(I\otimes Z_i)
    \end{aligned}
    ,
\end{equation}
where $\left< s_i \right>$ represent the stabilizers of the code, $\{X_{\ell_i}, Z_{\ell_i}\}$ are the logical operators, and $\{r_i, t_i\}$ are the gauge operators. We may represent all this information in the form of a stabilizer tableau, as seen in section \ref{sec:notation_conventions}.
Note that the Clifford encoding isometry is \emph{not unital} but instead maps 1-to-1 onto the subsystem center
\begin{equation}
    P\paulis(n)P=U \Bigl(\paulis(k)\otimes\paulis(g)\Bigr)U^\dagger \tabspace\text{with projector}\tabspace P=\prod_{i=1}^m \left(\frac{1+s_i}{2}\right).
\end{equation}
Such a Clifford isometry may be constructed as a Clifford encoding circuit,
such as for example the circuit depicted in Figure \ref{fig:subs_example}.
\footnote{Strictly speaking it suffices to note that any Clifford isometry (as a linear map)
may be written as ZX-diagram in the Clifford fragment,
but we include this step for ease of the reader.}
\begin{figure}
    \centering
    \includegraphics[width=0.5\linewidth]{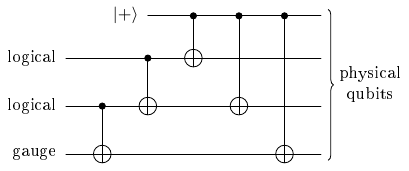}
    \caption{Example of an encoding circuit of a subsystem code.}
    \label{fig:subs_example}
\end{figure}
One may decompose such a Clifford circuit into a Clifford ZX-diagram,
which may then be further simplified via \cite[theorem 5.4]{duncan_graph-theoretic_2020} to some graph-like ZX-diagram,
with some leftover local Clifford gates on the input and output wires such as
\begin{equation}
    ~\vcenterinclude{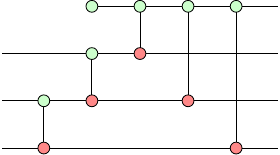}
    ~=~\vcenterinclude{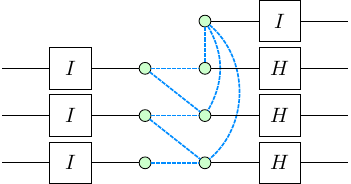}
    ~.
\end{equation}
As such every subsystem code arises as a graph-like ZX-diagram with logical and gauge input wires and physical output wires,
which leads us to the definition of a subsystem code in the \zxcalculus.
\begin{definition}[Subsystem code in the \zxcalculus]\label{def:zx-enc_graph}
    Given a subsystem code specified as a stabilizer tableau with a given choice of gauge, logical Paulis and a Clifford encoding isometry $U$ as in \eqref{eq:cliff_encoding_isometry},
    we define its ZX-encoding graph as any Clifford ZX-diagram realizing the Clifford encoding isometry:
    \begin{equation}
        \left(\vcenterinclude{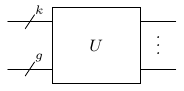}\right)
        =
        \left(\vcenterinclude{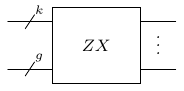}\right)
        ~.
    \end{equation}
    The ZX-diagram may always be simplified to a graph-like ZX-diagram \cite[theorem 5.4]{duncan_graph-theoretic_2020} with some Clifford gates on the input and output.
    Furthermore, given an order on physical qubits, a ZX-encoding graph is in (unique) normal form if it satisfies the conditions of corollary \ref{cor:norm_form_input}.
\end{definition}
We point out that this definition encompasses stabilizer codes as well, since the gauge qubits aren't needed to define a valid Clifford encoding isometry. Also, we will always reduce the ZX-encoding graphs to their graph-like ZX-diagram, and in particular, there will be no inner nodes. So, all the nodes in the graph will correspond to either an input or output qubit. In this context, we will refer to qubits as nodes and vice versa.

\subsection{Biadjacency matrix: Full rank}\label{subsec:zx-enc_full_rank}

As described in \cite[Section 6]{duncan_graph-theoretic_2020}, we may decompose the ZX-encoding graph further into a layer of CZ gates, followed by a CNOT circuit, a full Hadamard layer, and another final CZ layer
\begin{equation}
    \label{eq:decomposition-CZ-CNOT-H-CZ}
    \includegraphics[valign=c]{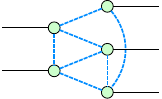}
    \tabspace = \tabspace
    \includegraphics[valign=c]{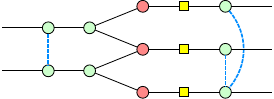}
    ~.
\end{equation}
More precisely, the middle layer defines a reversible classical linear circuit (a CNOT circuit) given by the biadjacency matrix between input and output nodes acting as
\begin{equation}
    \label{eq:classical-linear-circuit}
    \left(\includegraphics[valign=c]{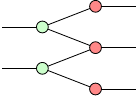}\right) 
    =
    \begin{pmatrix}
        1 & 0 \\
        1 & 1 \\
        0 & 1 \\
    \end{pmatrix}:
    \tabbspace
    \Ket{\begin{psmallmatrix}u\\v\end{psmallmatrix}}
    ~\mapsto~
    \Ket{\begin{psmallmatrix}
        1 & 0 \\
        1 & 1 \\
        0 & 1 \\
    \end{psmallmatrix}
    \begin{psmallmatrix}u\\v\end{psmallmatrix}}
    =\Ket{\begin{psmallmatrix}u\\u+v\\v\end{psmallmatrix}}.
\end{equation}
Indeed this may be easily seen by fixing the inputs with $\ket{s}=(X^s\vcenterinclude{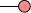})$ :
\begin{equation}
     \includegraphics[valign=c]{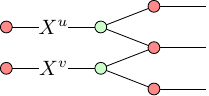}    
     \tabspace = \tabspace
    \includegraphics[valign=c]{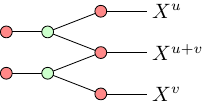}
    \tabspace = \tabspace
    \includegraphics[valign=c]{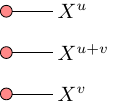}
    ~,
\end{equation}
and we get the expected mapping.
Since the ZX-encoding graph defines a Clifford isometry as a linear map, the following well-known property on the biadjacency matrix (similar to the unitary case) is entailed.

\begin{proposition}[Biadjacency matrix: Full rank]
    \label{prop:full_rank_graph}
    A graph-like ZX-diagram without interior spiders defines an isometry if and only if its biadjacency matrix between its input and output nodes is full rank.
    In particular, for any realization of a Clifford encoding isometry as a graph-like ZX-encoding graph without interior spiders, every input node is connected to at least one output node.
\end{proposition}
Since we couldn't find a reference, we give a proof for ease of the reader.
\begin{proof}
    Suppose the biadjacency matrix of the graph-like ZX-diagram is not full rank,
    then its underlying classical CNOT circuit \eqref{eq:classical-linear-circuit} is not reversible.
    Conversely suppose the biadjacency matrix is full rank, then the smaller computational basis is mapped collision-free into the larger computational basis.
    In particular the classical CNOT circuit \eqref{eq:classical-linear-circuit} considered as a quantum circuit defines an isometry, and thus also the entire ZX encoding circuit when including its surrounding unitary $\CZ$ layers \eqref{eq:decomposition-CZ-CNOT-H-CZ}.
\end{proof}

\subsection{ZX-encoding graphs: Normal form}\label{subsec:zx-enc_norm_form}

With this in hand, we now address the normal form of ZX-encoding graphs introduced by Hu and Khesin \cite{hu_improved_2022} for stabilizer states. As mentionned before, this result holds for ZX-encoding graphs through map--state duality. The normal form is a very useful tool to compare ZX-encoding graphs, and determine if they are equal. In particular, because of the equivalence between local Clifford operators and local complementations (see Section \hyperref[sec:zx-calculus]{2}), two ZX-encoding graphs may appear to be different when they are in fact equal. For example, consider the two following ZX-encoding graphs
\begin{equation}
    \mathrel{\vcenter{\hbox{
    \includegraphics{figures/stab_in_ZX/LC_equiv_initial.pdf}}}}
    \hspace{2mm} = \hspace{1mm}
    \mathrel{\vcenter{\hbox{
    \includegraphics{figures/stab_in_ZX/LC_equiv_final.pdf}}}} ~.
\end{equation}

To understand the normal form, we will look at these local Clifford gates for a given graph-like ZX-encoding graph.
For visual simplicity, we will work with the equivalent case of stabilizer states via map-state duality such as:
\begin{equation}
    \includegraphics[valign=c]{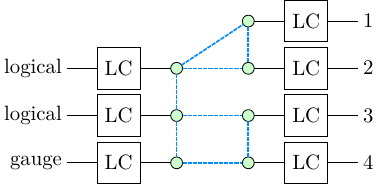}
    ~=~
    \includegraphics[valign=c]{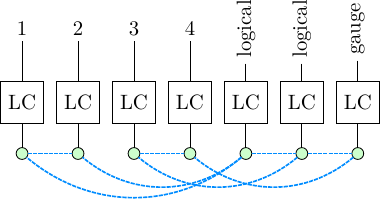}
    ~.
\end{equation}
We would like to remove now the following degree of freedom (as also in \cite{elliott_graphical_2008}):
\begin{equation}\label{eq:cleanup}
    \LCboxed = \{I, \sqrt{X}, \sqrt{Z}, \sqrt{X}\sqrt{Z}, \sqrt{Z}\sqrt{X}, H \}
    \tabspace\implies\tabspace
    \LCboxed = \{I, \sqrt{Z}, H\}, 
\end{equation}
with last combination of $\sqrt{Z}$ and $\sqrt{X}$ written as $\sqrt{Z}\sqrt{X}\sqrt{Z}=\sqrt{X}\sqrt{Z}\sqrt{X}=H$. We observe that one may also rewrite the following combination:
\begin{equation}
    \sqrt{Z}\sqrt{X} = \sqrt{X}(\sqrt{X}\sqrt{Z}\sqrt{X}) = \sqrt{X} H,
\end{equation}
such that all the $\sqrt{X}$ in the original set in equation \eqref{eq:cleanup} are applied to the ZX-diagram first. As a consequence, we may always remove these $\sqrt{X}$ gate via local complementation:
\begin{equation}
    \includegraphics[valign=c]{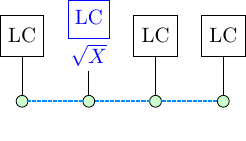}
    \tabspace = \tabspace
    \includegraphics[valign=c]{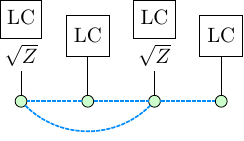}~.
    \vspace{-\baselineskip/2}
\end{equation}
We may perform such a cleanup on all the vertices and in any order. Indeed, this does not undo the progress on the already cleaned up vertices since
\begin{equation}
    I\mapsto I\sqrt{Z}=\sqrt{Z},
    \tabbspace
    \sqrt{Z}\mapsto \sqrt{Z}\sqrt{Z}=I,
    ~\text{and }~ 
    H\mapsto \sqrt{Z}H=\sqrt{X}\sqrt{Z},
\end{equation}
which would at worst enable us to remove yet another $\sqrt{X}$. By recursively replacing $\sqrt{X}$ via local complementation, we can clear up the first degree of freedom as originally proven in \cite{elliott_graphical_2008} and observed in \cite[Section 6]{backens_zx-calculus_2014}.
Further, we may remove any neighboring pair of Hadamard gates via edge complementation:
\begin{equation}
    \includegraphics[valign=c]{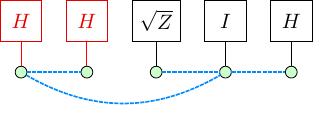}
    \tabspace = \tabspace
    \includegraphics[valign=c]{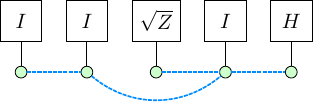}~.
\end{equation}
We are now left with one last degree of ambiguity, i.e. when it is not possible to cancel the remaining Hadamard gates:
\begin{equation}
    \includegraphics[valign=c]{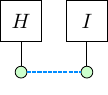}
    ~=~
    \includegraphics[valign=c]{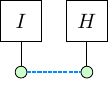}
    ~,
    \tabbspace
    \mathand
    \tabbspace
    \includegraphics[valign=c]{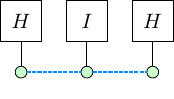}
    ~=~
    \includegraphics[valign=c]{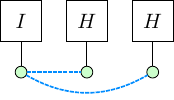}~.
\end{equation}
To resolve this issue, Hu and Khesin \cite{hu_improved_2022} introduced an idea analogous to graph isomorphism testing via canonization: we may simply choose some order on qubits (any order will do), and require that qubits with a Hadamard gate may only be connected to lower-index qubits.
If not yet satisfied, we may simply move Hadamard gates along some existing edge to a higher-index qubit. For example, we may choose to order the qubits from left to right in the ZX-diagram below:
\begin{equation}
     \includegraphics[valign=c]{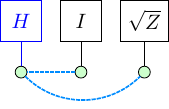}
    =
    \includegraphics[valign=c]{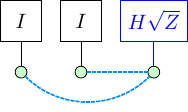}
    =
    \includegraphics[valign=c]{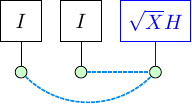}
    =
    \includegraphics[valign=c]{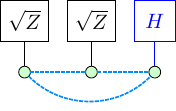},
\end{equation}
where we used a final local complementation to remove the remaining $\sqrt{X}$.\\
This ordering along with the previous reductions define the desired normal form.
More precisely, for any stabilizer state there exists a unique such decomposition. We point out that this normal form may be called the GSLC form, for graph-state with local Clifford \cite[definition 5.3.6]{kissinger_picturing_2024}.

\begin{theorem}[{\cite[theorem III.4]{hu_improved_2022}: normal form}]
    \label{thm:normal_form}
    Given an order on qubits, the normal form for stabilizer states as described above is unique in the sense where,
    \begin{equation}
        \Ket{\graph + \LC} = \Ket{\graph' + \LC}
        \tabspace\implies\tabspace
        \graph = \graph'
        \tabspace\mathand\tabspace
        \LC\equiv\LC'.
    \end{equation}
    In particular, the normal form of ZX-encoding graphs is unique.
\end{theorem}

We give a simplified proof inspired by a simple counting argument for stabilizer codes, which we found in a later version of \cite[Section II.A]{khesin_universal_2025}. 
Note that an equivalent version of this theorem was found in \cite[theorem 5.5.4]{kissinger_picturing_2024} which may be also proven the same way.

\begin{proof}
    Recall that any stabilizer state admits at least one graph with local Clifford gates given in normal form.
    As such there are at least as many graphs in normal form as stabilizer states.
    
    The number of stabilizer states is (see \cite[theorem 20]{gross_hudsons_2006} or  \cite[proposition 2]{aaronson_improved_2004})
    \begin{align}
        \frac{
                \left(2^{2n}-1\vphantom{2^{\mathstrut}}\right)
                \left(2^{2n-1}-2\vphantom{2^{\mathstrut}}\right)
                \cdots
                \left(2^{n+1}-2^{n-1}\right)
            }
            {
                \left(2^n-1\vphantom{2^{\mathstrut}}\right)
                \left(2^n-2\vphantom{2^{\mathstrut}}\right)
                \cdots
                \left(2^n-2^{n-1}\vphantom{2^{\mathstrut}}\right)
            }
        =& \prod_{q=0}^{n-1}
        \frac{
                \left(2^{n-q}+1\right)\cancel{\left(2^{n-q}-1\right)}
            }
            {
                \cancel{\left(2^{n-q}-1\right)}
            }, \\
        =& \prod_{q=1}^{n}\left(2^{n-(q-1)}+1\right),
    \end{align}
    where we count the number of generating sets divided by the number of equivalent generating sets for a given stabilizer state. Note that we work in the symplectic representation, so we count the number of isotropic subspaces.
    
    Now we aim to count the number of graphs with local Clifford gates in the normal form. 
    The number of ways a qubit may be connected to qubits with a higher index is given by:
    \begin{equation}
        \left| q\to n, q\to n-1, \ldots, q\to q+1\right|=\left(2^{n - q}-1\right), \text{ where }(1\leq q \leq n).
    \end{equation}
    For a given graph connectivity, either the qubit is not connected to a lower indexed qubit or it is, in which case it cannot hold a Hadamard gate. So, considering all possible conntectivities to lower indexed qubits and all possible gates, we get the following number of combinations for a one qubit in the graph:
    \begin{equation}
        \Bigl( 1\times|\LC=1/\sqrt{Z}/H| + (2^{n-q}-1)\times|\LC=1/\sqrt{Z}| \Bigr)
        = \Bigl(3 + (2^{n-q}-1)\times 2\Bigr) = (2^{n-(q-1)}+1).
    \end{equation}
    This is independent for every qubit. By taking a product of all the qubits, we span the whole set of graph states with local Clifford gates in normal form, and it agrees with the number of stabilizer states. Consequently, graph states with local Clifford gates and stabilizer states form a bijection.
    Via map--state duality, we conclude that ZX-encoding graphs, or Clifford isometries, in the normal form also map in a 1--to--1 fashion onto stabilizer codes.
\end{proof}

Before proceeding, we'll note a specific normal form for ZX-encoding graphs: recall that for Clifford isometries any input node is connected with at least one output node (proposition \ref{prop:full_rank_graph}).
So we may move any leftover Hadamard towards output nodes via pivoting:
\begin{equation}
    \includegraphics[valign=c]{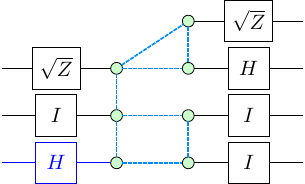}
    \tabspace = \tabspace
    \includegraphics[valign=c]{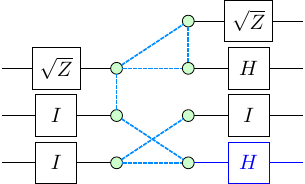}~.
\end{equation}
This choice is equivalently captured by simply choosing the input nodes to be indexed with negative integers, and output nodes, with positive integers. Then, the normal form will give the same ZX-encoding graph. 

This form is useful when doing gauge fixing, as a $H$ gate on the gauge input would turn a Z gauge fixing into a X gauge fixing. However, there could still be a $\sqrt{Z}$ gate on the gauge input, which would turn the X gauge fixing into the Y gauge fixing and vice versa. While the Z gauge operator corresponds to a stabilizer promoted to operator, the X and Y gauge fixings are "made up" degrees of freedom, so swaping them isn't an issue.

\begin{corollary}[Normal form: input wires] \label{cor:norm_form_input}
    For the normal form of a Clifford isometry, we may assume its input wires (logical and gauge) to contain only $\LC=\{I, \sqrt{Z}\}$.
\end{corollary}

\subsection{GSLC normal form = AP normal form}\label{subsec:gslc_ap_norm_form}

For completeness, let us draw the connection between above normal form (by Hu and Khesin), the affine with phases (AP) form by Dehaene and De Moor \cite{dehaene_clifford_2003}, and its AP normal form by Poór, Kissinger and van de Wetering in \cite{poor_unique_2022} and \cite{kissinger_picturing_2024}.
\begin{definition}[\cite{kissinger_picturing_2024}]
    We say a graph-like Clifford diagram is in affine with phases form (AP form) when:\\
    \begin{enumerate}
        \item every boundary spider is connected to exactly one input or output,
        \item every internal spider has a phase of 0 or $\pi$, and
        \item no two internal spiders are connected to each other.
    \end{enumerate}
\end{definition}
The \emph{affine} block consists of internal spiders because they form an affine subspace, while the \emph{phase} block consists of the spiders connected to the inputs or outputs because it corresponds to a diagonal unitary matrix.

The states associated via map--state dulality with ZX-diagrams in the AP normal form are proportionnal to state vectors determined by the connectivity of the ZX-diagram and the spider's phases:
\begin{equation}
    \sum_{ \vec{x} \in \text{rows}(A)} e^{i\phi(\vec{x})} \ket{\vec{x}}.
\end{equation}
The rows of $A$ span the affine subspace in $\mathbb{F}_2$, whose vectors correspond to the connectivity of the red spiders. The phase function $\phi(\vec{x})$ is built by summing the effect of Clifford gates in the phase block: the $j$th spider with a phase $\alpha$ will add the term $\alpha x_j$, while a $CZ$ gate, i.e. a Hadamard edge, between the $j$th and $k$th spider adds the term $\pi x_j x_k$. We refer the reader to section 5.3 of \cite{kissinger_picturing_2024} for further explanation of this form.

In particular the GSLC normal form and the AP normal form are literal identical forms, two sides of one and the same coin.
Consider for instance the following stabilizer state in normal form (by Hu and Khesin \cite{hu_improved_2022}):
\begin{equation}
    \includegraphics[valign=c]{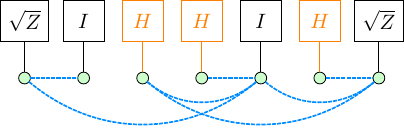}~.
    \vspace*{-\baselineskip/2}
\end{equation}
One may decompose the graph into the following two layers (see also \cite[remark 5.5.8]{kissinger_picturing_2024}):
\vspace*{-\baselineskip/2}
\begin{equation}
    \includegraphics[valign=c]{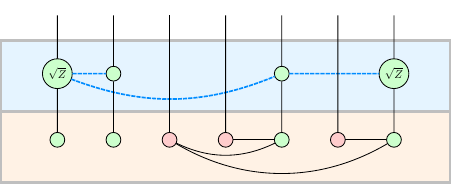}~,
    \vspace*{\baselineskip/2}
\end{equation}
where the orange and blue boxes highlight respectively the affine and phase blocks of the form.
As such the linear constraint in the AP normal form
may be read off as the layer given originally by Hadamard nodes, each defining an X-spider as a pivot \footnote{Since we work in the symplectic representation, the linear constraint will be always homogoneous.}:
\begin{equation}
    A = \begin{pmatrix}
        0 & 0 & 1 & 0 & 1 & 0 & 1 \\
        0 & 0 & 0 & 1 & 1 & 0 & 0 \\
        0 & 0 & 0 & 0 & 0 & 1 & 1 
    \end{pmatrix}
    \tabspace
    \mathand
    \tabspace
    b = \begin{pmatrix}
        0\\0\\0
    \end{pmatrix}
    ~,
\end{equation}
while the phase polynomial amounts to the remaining layer of diagonal Clifford gates:
\begin{equation}
    \phi(x_1,\ldots,x_7) = \frac{\pi}{2} \Bigl(x_1 + x_7\Bigr)+ \pi \Bigl(x_1x_2 + x_1x_5 + x_5x_7\Bigr)\tabbspace (x_i=0,1).\hspace*{-1cm}
\end{equation}
As a consequence, our simplified proof of theorem \ref{thm:normal_form} (the normal form by Hu and Khesin)
carries over to the AP normal form, and thus simplifies the proof of \cite[theorem 5.5.4]{kissinger_picturing_2024}. It all amounts to a simple counting argument.

\subsection{Constructing ZX-encoding graphs} \label{subsec:const_zx-enc}
While it is straightforward to find the stabilizer generators and logical operators of the stabilizer code represented by a given ZX-encoding graph, the reverse may not be obvious. In fact, one can find the stabilizers, the logical operators, and gauge operators by first transcribing them into a row matrix in the symplectic representation $ M = (G_X|G_Z)$,
and then performing a Gaussian elimination along with additional column operations (based on local Cliffords and qubit permutations). 
Van den Nest et al. \cite{nest_graphical_2004} proved that the resulting square matrix on the right-hand side must be symmetric, and thus correspond to the adjancency matrix of a graph. The result was found to convert stabilizer states into a graph state, but we may use map--state duality to apply it to stabilizer codes as Clifford isometries.

Every row in $M$ represents a stabilizer, a logical operator or a gauge operator, and $G_X$ ($G_Z$) is their $X$ ($Z$) support. As such, both $G_X$ and $G_Z$ are square matrices with entries in $\mathbb{F}_2$. To illustrate this, we turn to our trusted 3-qubit code example. To create a telling example, we will start from a particular ZX-encoding graph in its LC orbit, from which we extract the vertex stabilizers as seen in the background section of the Letter:
\begin{equation}\label{eq:zx-enc_const}
    \begin{split}
        \mathrel{\vcenter{\hbox{ \includegraphics{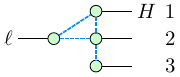}}}}
        ~ \rightarrow 
        v = \{ v_1=(Z | ZZI),& v_2 = (Z|XXZ), \\ & v_3=(I|IZX), v_l=(X|XZI) \},
    \end{split}
\end{equation}
where the support on the vertex stabilizers is ordered as follows $v_i= (\ell | 1 2 3)$. From these, we derive the stabilizers $\left< YYZ, IXZ\right>$, and the logical operators $X_\ell = XZI$ \& $Z_\ell = ZZI$. Using this information, we build the row matrix:
\begin{equation}\label{eq:gen_matrix}
    M = (G_X | G_Z ) =
    \begin{pNiceMatrix}[first-row,first-col]
           & \ell & 1 & 2 & 3 & \ell & 1 & 2 & 3  \\
        \overline X_\ell & 1 & 1 & 0 & 0 & 0 & 0 & 1 & 0 \\
        \overline Z_\ell & 0 & 0 & 0 & 0  & 1 & 1 & 1 & 0 \\
        s_1 & 0 & 1 & 1 & 0  & 0 & 1 & 1 & 1\\
        s_2 & 0 & 0 & 0 & 1 & 0 & 0 & 1 & 0
    \end{pNiceMatrix}
    ~.
\end{equation}

The next step is to perform the Gaussian elimination such that our final row matrix will be of the form $(I | A)$, $A$ being the adjacency matrix of the ZX-encoding graph we are looking for. In this form, every row will correspond to a vertex stabilizer. While the stabilizers and logical operators can be found from a ZX-encoding graph by composing vertex stabilizers, performing row operations on the stabilizer operators and logical operators will untangle the vertex stabilizers. 

However, a standard gaussian elimination may not be sufficient to bring the matrix into the desired form. In the example shown in equation \eqref{eq:gen_matrix}, the second row is empty on the left-hand side, making it impossible to bring $G_X \mapsto I$. To solve this, we point out that we may choose another encoding in the LC of that stabilizer code by performing local Clifford gates and qubit permutations.
The effect of a Hadamard gate on a qubit will be to swap that qubit's column on the $X$ side and on the $Z$ side:
\begin{equation}\label{eq:gen_matrix_H1}
    M \overset{H_1}{=\joinrel=}
    \begin{pmatrix}
        1 & 0 & 0 & 0 & \vrule & 0 & 1 & 1 & 0 \\
         0 & 1 & 0 & 0 & \vrule & 1 & 0 & 1 & 0 \\
         0 & 1 & 1 & 0 & \vrule & 0 & 1 & 1 & 1\\
         0 & 0 & 0 & 1 & \vrule & 0 & 0 & 1 & 0
    \end{pmatrix}
  ~.
\end{equation}
Now, we have a matrix on the left-handside that can be brought to $I$ through gaussian elimination:
\begin{equation}\label{eq:gen_matrix_H1_a}
     \overset{s_1+\overline Z_\ell}{=\joinrel=\joinrel=}
    \begin{pmatrix}
        1 & 0 & 0 & 0 & \vrule & 0 & 1 & 1 & 0 \\
         0 & 1 & 0 & 0 & \vrule & 1 & 0 & 1 & 0 \\
         0 & 0 & 1 & 0 & \vrule & 1 & 1 & 0 & 1\\
         0 & 0 & 0 & 1 & \vrule & 0 & 0 & 1 & 0
    \end{pmatrix}
    = (I | A)
  ~.
\end{equation}
The matrix on the right-hand side corresponds to the adjacency matrix of the following ZX-encoding graph:
\begin{equation}\label{eq:zx-enc_const_final}
    \mathrel{\vcenter{\hbox{ \includegraphics{figures/gauge_fixing/S_one_Example.pdf}}}}
    ~.
\end{equation}
We note that in general, the diagonal entries may be non-zero. In that case, they can be ignored, as there are no self-edges in graph states \cite{nest_graphical_2004}.

The Hadamard gate on the first physical qubit found in the original ZX-encoding graph of equation\eqref{eq:zx-enc_const} disappeared in our final result in equation \eqref{eq:zx-enc_const_final} above. This is because we had to change the encoding in order to perform the gaussian elimination procedure. If one wants to find exactly the same encoding as the original stabilizer code, then it is possible to compute the LC orbit of the ZX-encoding graph found, and choose the encoding which matches the stabilizers and logical operators.

The other operations available to bring $G_X$ to $I$ are to perform a phase gate $\sqrt{Z}$ on a qubit. Since this maps $X$ to $Y=iZX$, the effect is to add the qubit's column on the $X$ side to the one on the $Z$ side. Finally, qubit permutations swap columns within $G_X$ and $G_Z$. It is important to keep track of the qubit order in order to retrieve the ZX-encoding graph with the correct input and output nodes. All these operations, including the standard row operations, may be performed in any order.
\section{Clifford Conversion < Gauge Fixing}\label{sec:clifford_conversion}

We now take a closer look at Clifford conversion as was considered in \cite{hill_fault-tolerant_2013} and \cite{hwang_fault-tolerant_2015},
in the sense of converting between stabilizer codes via a series of 2-qubit entangling Cliffords
\[
    \includegraphics[valign=c,height=1.25cm]{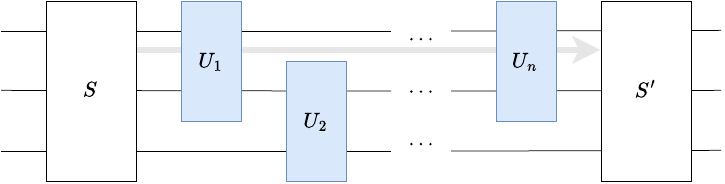}
\]
from the perspective of ZX encoding graphs,
and its relation to code switching via a sequence of gauge check measurements as in \cite{colladay_rewiring_2018} and \cite{huang_transversal_2018}.
Both define promising methods in current research on fault-tolerant quantum computation
(\cite{butt_fault-tolerant_2024,heusen_measurement-free_2024,postler_demonstration_2024,heusen_efficient_2024}
and \cite{ryan-anderson_realization_2021,ryan-anderson_implementing_2022})
each with their own advantages from a fault tolerance perspective.

We will find that one may always equally realize the conversion by any single 2-qubit entangling Clifford as a single gauge check measurement (up to Pauli correction)\footnote{Equality in the symplectic representation and thus up to Pauli correction.}
\[
    \includegraphics[valign=c,height=1cm]{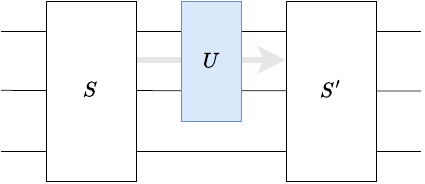}
    ~=~
    \includegraphics[valign=c,height=1cm]{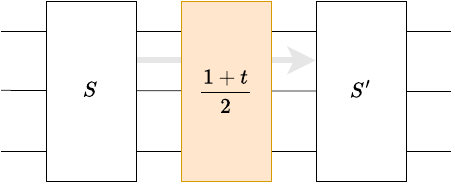}
\]
and as such also any \emph{sequential conversion} via 2-qubit entangling Cliffords via a series of gauge check measurements.
This is in accordance with \cite[prop.~2.2]{kliuchnikov_stab_circ_2023}
which asserts instead that any $\pi/4$ rotation may be equally realized as a gauge check measurement:
\[
    \includegraphics[valign=c,height=1cm]{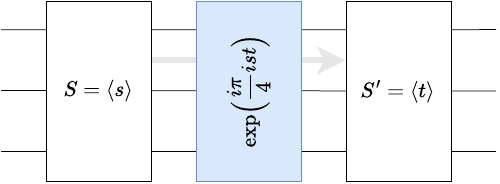}
    ~=~
    \includegraphics[valign=c,height=1cm]{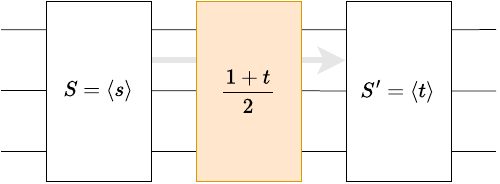}
\]
Conversely, we will unravel how transvections (as one unitary implementation of gauge fixing) may be understood as a conversion by some highly entangling Clifford.
Finally, we will provide a gauge fixing example (given by code deformation)
that cannot be arranged as a single Clifford conversion step via a single 2-qubit entangling Clifford.

Let us begin with a precise setting of Clifford conversion
as considered in \cite{hill_fault-tolerant_2013} and \cite{hwang_fault-tolerant_2015}.
For this we note that the 2-qubit Clifford group allows the following decomposition:

\begin{proposition}
    The Clifford group on 2-qubits admits the decomposition:
    \begin{align}
        \cliffords(n=2)
        &= \{\SWAP,\identity\}\circ\Bigl(\LC\otimes\LC\circ\{\identity,\CZ\}\circ\LC\otimes\LC\Bigr) \notag\\
        &=
        \left\{
            \includegraphics[valign=c,scale=0.35]{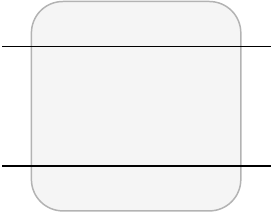},
            \includegraphics[valign=c,scale=0.35]{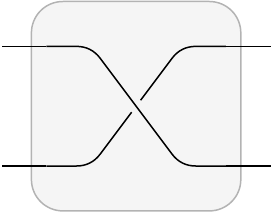}
            \vphantom{\includegraphics[valign=c,scale=0.4]{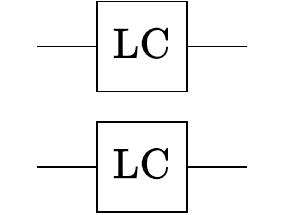}}
        \right\}
        \circ
        \left\{
            \hspace{-2mm}\includegraphics[valign=c,scale=0.4]{figures/cliff_conv/2q_cliff_LCS_only.pdf}\hspace{-2mm},~
            \hspace{-2mm}\includegraphics[valign=c,scale=0.4]{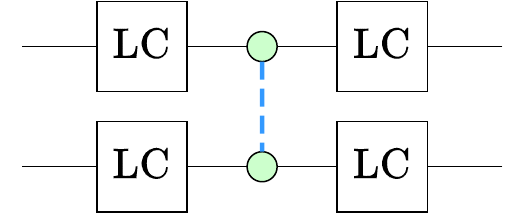}\hspace{-1mm}
        \right\}
    \end{align}
    In particular, any 2-qubit entangling Clifford may be simplified to a circuit with a single $\CZ$ gate surrounded by local Cliffords.
    For example, the following combination of entangling Cliffords may be written as a single $\CZ$ gate up to local Cliffords,
    \[
        U = \CZ\circ\CNOT = \left(\sqrt{Z}\otimes \sqrt{X}\right) \circ \CZ \circ \left(I\otimes \sqrt{X}\right).
    \]
\end{proposition}
\begin{remark}
    The proposition may be also easily proven via stabilizer tableaus.\\
    For illustration we however provide a proof via ZX encoding graphs.
\end{remark}

\begin{proof}
Recall that every Clifford unitary (as a special instance of Clifford isometries) may be written as a ZX encoding graph (section \ref{sec:subs_zx})
\[
    \includegraphics[scale=0.45]{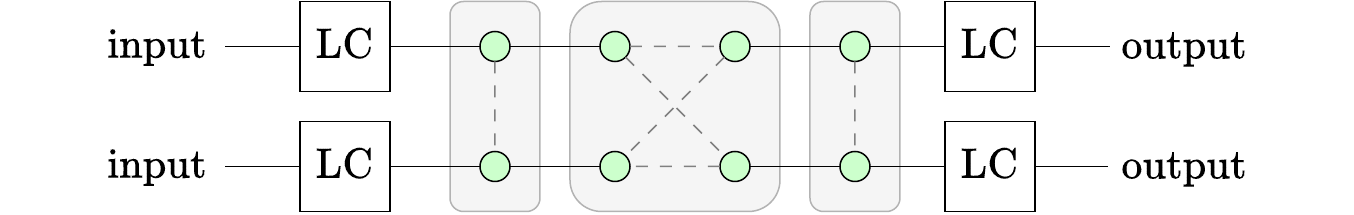}
\]
with possible entangling $\CZ$ gate between the input and between the outputs, and a full rank bipartite graph
\begin{align*}
    \includegraphics[valign=c,scale=0.33]{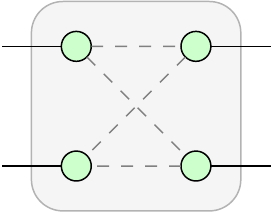}
    &\in
    \left\{
        \includegraphics[valign=c,scale=0.33]{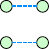},
        \includegraphics[valign=c,scale=0.33]{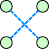},
        \includegraphics[valign=c,scale=0.33]{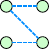},
        \includegraphics[valign=c,scale=0.33]{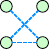},
        \includegraphics[valign=c,scale=0.33]{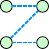},
        \includegraphics[valign=c,scale=0.33]{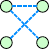}
    \right\}
    \\
    &=
    \left\{
        \includegraphics[valign=c,scale=0.33]{figures/cliff_conv/2q_cliff_triv.pdf},
        \includegraphics[valign=c,scale=0.33]{figures/cliff_conv/2q_cliff_swap.pdf}
    \right\}
    \circ
    \left\{
        \includegraphics[valign=c,scale=0.33]{figures/cliff_conv/2q_cliff_graph_1a.pdf},
        \includegraphics[valign=c,scale=0.33]{figures/cliff_conv/2q_cliff_graph_2a.pdf},
        \includegraphics[valign=c,scale=0.33]{figures/cliff_conv/2q_cliff_graph_3a.pdf}
    \right\}.
\end{align*}
Along with entangling $\CZ$ gates, we thus aim to find all of these ZX encoding graphs (up to possible swap):
\begin{align}
    \notag
    &\left\{
        \includegraphics[valign=c,scale=0.3]{figures/cliff_conv/2q_cliff_triv.pdf},
        \includegraphics[valign=c,scale=0.3]{figures/cliff_conv/2q_cliff_swap.pdf}
    \right\}
    \circ
    \left\{
        \includegraphics[valign=c,scale=0.3]{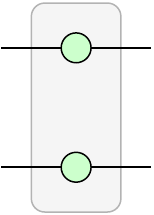},
        \includegraphics[valign=c,scale=0.3]{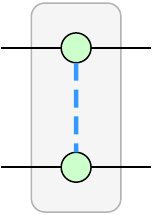}
    \right\}
    \circ
    \left\{
        \includegraphics[valign=c,scale=0.3]{figures/cliff_conv/2q_cliff_graph_1a.pdf},
        \includegraphics[valign=c,scale=0.3]{figures/cliff_conv/2q_cliff_graph_2a.pdf},
        \includegraphics[valign=c,scale=0.3]{figures/cliff_conv/2q_cliff_graph_3a.pdf},
    \right\}
    \circ
    \left\{
        \includegraphics[valign=c,scale=0.3]{figures/cliff_conv/2q_cliff_cz_none.pdf},
        \includegraphics[valign=c,scale=0.3]{figures/cliff_conv/2q_cliff_cz_present.pdf}
    \right\}
    \\
    \label{eq:2q_cliff_all_graphs}
    &=
    \left\{
        \includegraphics[valign=c,scale=0.3]{figures/cliff_conv/2q_cliff_triv.pdf},
        \includegraphics[valign=c,scale=0.3]{figures/cliff_conv/2q_cliff_swap.pdf}
    \right\}
    \circ
    \left\{
        \includegraphics[valign=c,scale=0.3]{figures/cliff_conv/2q_cliff_graph_1a.pdf},
        \includegraphics[valign=c,scale=0.3]{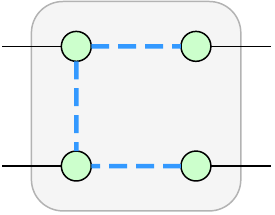},
        ~\ldots~,
        \includegraphics[valign=c,scale=0.3]{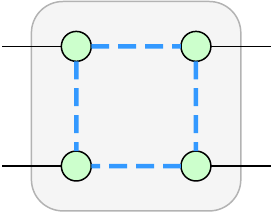},
        \includegraphics[valign=c,scale=0.3]{figures/cliff_conv/2q_cliff_graph_2a.pdf},
        \includegraphics[valign=c,scale=0.3]{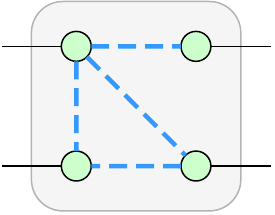},
        ~\ldots~,
        \includegraphics[valign=c,scale=0.3]{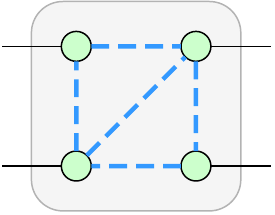}
    \right\}
\end{align}
Our aim is now to find all these graphs \eqref{eq:2q_cliff_all_graphs} starting from either of these expressions (plus some possible swap):
\[
    \hspace{-2cm}
    \left\{
        \includegraphics[valign=c,scale=0.35]{figures/cliff_conv/2q_cliff_triv.pdf},
        \includegraphics[valign=c,scale=0.35]{figures/cliff_conv/2q_cliff_swap.pdf}
        \vphantom{\includegraphics[valign=c,scale=0.4]{figures/cliff_conv/2q_cliff_LCS_only.pdf}}
    \right\}
    \circ
    \left\{
        \hspace{-2mm}\includegraphics[valign=c,scale=0.4]{figures/cliff_conv/2q_cliff_LCS_only.pdf}\hspace{-2mm},~
        \hspace{-2mm}\includegraphics[valign=c,scale=0.4]{figures/cliff_conv/2q_cliff_CZ_LCS.pdf}\hspace{-1mm}
    \right\}
\]
To better organize the proof, recall that we may freely move any phase gate through the Z-spider
\[
    \includegraphics[valign=c]{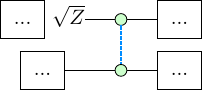}
    ~=~
    \includegraphics[valign=c]{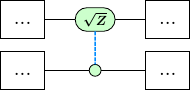}
    ~=~
    \includegraphics[valign=c]{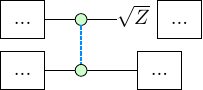}
\]
and that the single qubit Clifford group reads
\[
    \LC=\left\{I,\sqrt{Z},H,\sqrt{Z}H,H\sqrt{Z},\sqrt{Z}H\sqrt{Z}=H\sqrt{Z}H\right\}.
\]
As such we need to consider only the cases (plus a possible SWAP gate)
\[
    \left\{
        \includegraphics[valign=c]{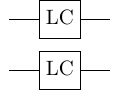},~
        \includegraphics[valign=c]{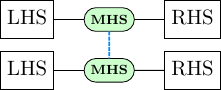}~~
    \right\}
\]
with lefthand side, middle, and righthand side
\[
    \LHS = \left\{I,H,H\sqrt{Z}\right\},
    \tabspace
    \MHS = \left\{I,\sqrt{Z}\right\},
    \tabspace
    \RHS = \left\{I,H,\sqrt{Z}H\right\}
\]
with them read as $H\sqrt{Z}:\ket{\phi}\mapsto H\sqrt{Z}\ket{\phi}$.
As a first counting argument we obtain
\begin{align*}
    &|\{\SWAP,\identity\}|\times\left(|\LC|^2 + |\LHS|^2\cdot|\mathrm{MHS}|^2\cdot|\RHS|^2\right)\\
    &= 2\times\left(6^2 + 3^2\cdot 2^2\cdot 3^2\right)
    = 720 = |\symplecticgroup(2n=4)|
\end{align*}
and as such we are on the right track.
There may be however some remaining cancellations.
Our strategy is thus to bring them all into normal form (corollary \ref{cor:norm_form_input}).
For the case of non-entangling Cliffords we note that we may factor the local Clifford group into
$\boxed{~\IN~}:=\{I,\sqrt{Z}\}$ and $\boxed{\OUT}:=\{I,\sqrt{Z},H\}$ via
\begin{gather*}
    \boxed{\OUT\bigstrut}\circ H\circ \boxed{~\IN~\bigstrut}
    = \{I,\sqrt{Z},H\} \circ H\circ \{I,\sqrt{Z}\} \\
    = \{H,H\sqrt{Z},\sqrt{Z}H,\sqrt{Z}H\sqrt{Z},I,\sqrt{Z}\}=\LCboxed
\end{gather*}
and as such we obtain the normal form for non-entangling Cliffords (plus possible swap):
\[
    \includegraphics[valign=c]{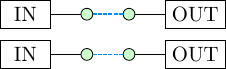}
    \tabbspace
    \includegraphics[valign=c]{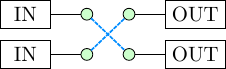}
\]
For the entangling cases we use the simple regrouping
\[
    \boxed{\MHS\bigstrut}\circ \boxed{\LHS\bigstrut}
    = \{I,\sqrt{Z}\}\cup \{I,\sqrt{Z}\}H\{I,\sqrt{Z}\}
    = \boxed{~\IN~\bigstrut}~\cup \left(~\boxed{\MHS\bigstrut}\circ H\circ \boxed{~\IN~\bigstrut}~\right)
\]
along with the decomposition:
\begin{gather*}
    \boxed{\RHS\bigstrut} = \{I,H,\sqrt{Z}H\} = \{H,I,\sqrt{Z}\}\circ H = \boxed{\OUT\bigstrut}\circ H\\
    \boxed{\RHS\bigstrut}\circ\boxed{\MHS\bigstrut}
    = \{I,H,\sqrt{Z}H,\sqrt{Z},H\sqrt{Z},\sqrt{Z}H\sqrt{Z}\} = \LCboxed
\end{gather*}
Using above regrouping and decomposition we thus obtain the following four cases:
\[
    \includegraphics[valign=c]{figures/cliff_conv/2q_cliff_entangling.pdf}
    ~:\tabspace
    \begin{array}{cc}
        \includegraphics[valign=c]{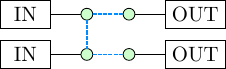}~,
        &
        ~\includegraphics[valign=c]{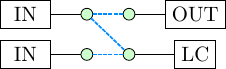}
        \\[7\jot]
        \includegraphics[valign=c]{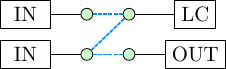}~,
        &
        ~\includegraphics[valign=c]{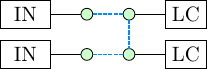}
    \end{array}
\]
Note that we may further refactor the local Clifford group into
\[
    \LCboxed
    = \{I,\sqrt{Z},\sqrt{X}\sqrt{Z}\sqrt{X}=H\}\circ\{I,\sqrt{X}\}
    = \boxed{\OUT\bigstrut} \circ \{I,\sqrt{X}\}.
\]
With this we obtain for the second case (and analogously for the third)
\begin{gather*}
    \includegraphics[valign=c]{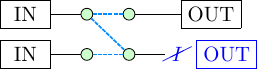}~,
    \tabbspace
    \includegraphics[valign=c]{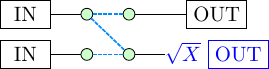}
    ~=~
    \includegraphics[valign=c]{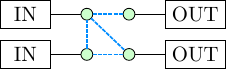}
\end{gather*}
with trivial first case $\boxed{\OUT}\circ I=\boxed{\OUT}$ and local complementation for the case $\boxed{\OUT}\circ\sqrt{X}$ (while absorbing any phase gate via $\sqrt{Z}\circ\boxed{~\IN~}=\boxed{~\IN~}$).
On the other hand we may also refactor the local Clifford group as
\[
    \LCboxed
    = \{I,\sqrt{Z}\} \circ \left\{I,\sqrt{X},\sqrt{Z}\sqrt{X}\sqrt{Z}=H\right\}
    = \boxed{~\IN~\bigstrut} \circ \{I,\sqrt{X},H\}.
\]
As such we obtain for the fourth case one the following three subcases
\begin{gather*}
    \includegraphics[valign=c]{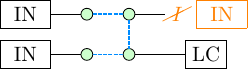}~,
    \tabbspace
    \includegraphics[valign=c]{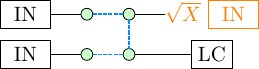}
    ~=~
    \includegraphics[valign=c]{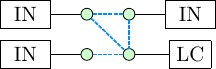}
    \\[2\jot]
    \includegraphics[valign=c]{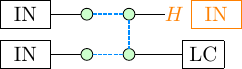}
    ~=~
    \includegraphics[valign=c]{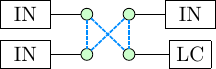}
\end{gather*}
where we used local complementation/pivoting as above, along with $\boxed{\LC}=\boxed{\LC}\circ\sqrt{X}=\boxed{\LC}\circ H$.
Using once more the decomposition $\boxed{\LC} = \boxed{\OUT} \circ \{I,\sqrt{X}\}$, we obtain via another local complementation:
\begin{gather*}
    \includegraphics[valign=c]{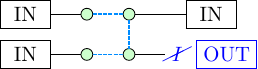}~,
    \tabbspace
    \includegraphics[valign=c]{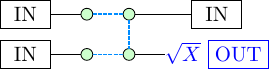}
    ~=~
    \includegraphics[valign=c]{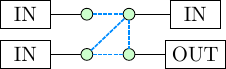}
    \\[2\jot]
    \includegraphics[valign=c]{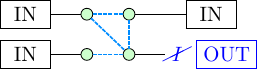}~,
    \tabbspace
    \includegraphics[valign=c]{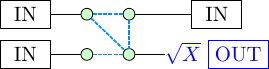}
    ~=~
    \includegraphics[valign=c]{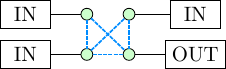}
    \\[2\jot]
    \includegraphics[valign=c]{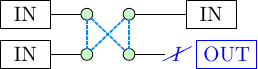}~,
    \tabbspace
    \includegraphics[valign=c]{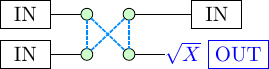}
    ~=~
    \includegraphics[valign=c]{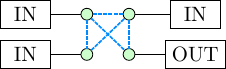}
\end{gather*}
For the last three subcases we may apply some final swap:
\begin{align*}
    \includegraphics[valign=c,scale=0.4]{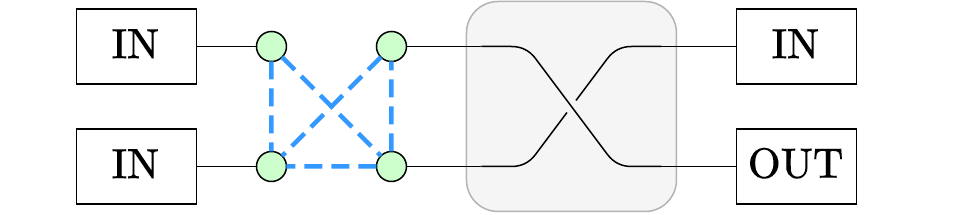}
    \hspace{-4mm}
    &=
    \hspace{-3mm}
    \includegraphics[valign=c,scale=0.4]{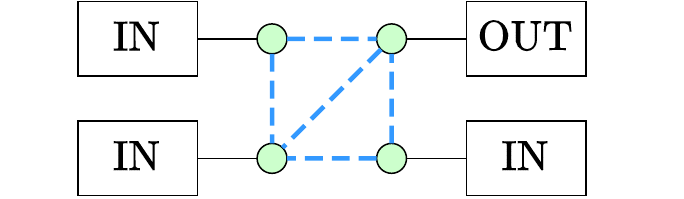}
    \\
    \includegraphics[valign=c,scale=0.4]{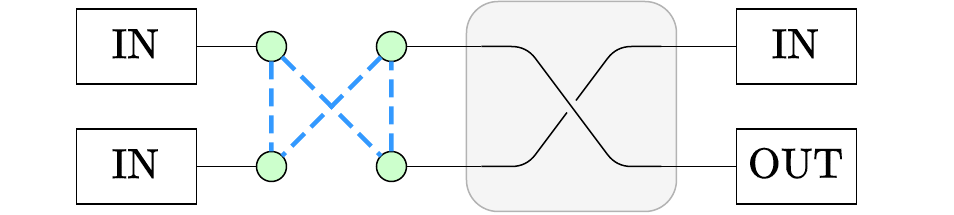}
    \hspace{-4mm}
    &=
    \hspace{-3mm}
    \includegraphics[valign=c,scale=0.4]{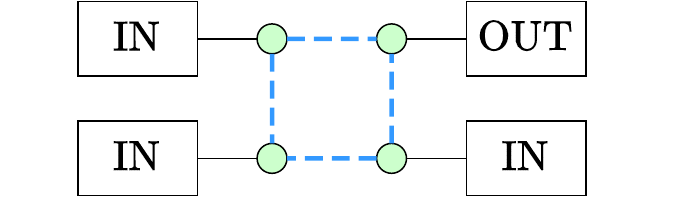}
    \\
    \includegraphics[valign=c,scale=0.4]{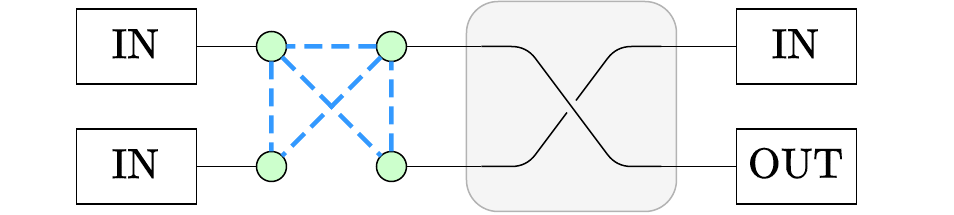}
    \hspace{-4mm}
    &=
    \hspace{-3mm}
    \includegraphics[valign=c,scale=0.4]{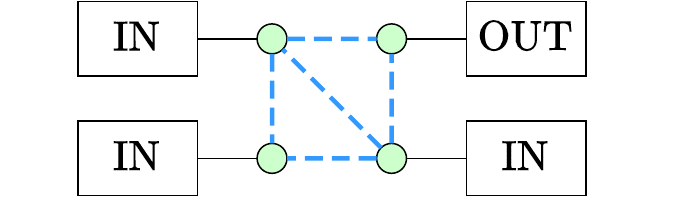}
\end{align*}
Finally note that we may freely add any remaining local Clifford from the left or right to any of the graphs above.\\
We thus found each of the ZX encodings graphs from \eqref{eq:2q_cliff_all_graphs} and the theorem is proven.
\end{proof}

With this at hand, we now separate between 2-qubit entangling gates and an optional swap of qubits.
This serves as an important distinction on fault tolerant implementations,
since permutation of qubits come as native operation only in certain architectures.
More precisely, we consider any 2-qubit entangling gate of the form (excluding swaps):
\[
    \includegraphics[valign=c]{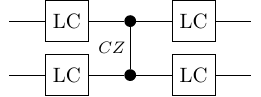}
    \tabspace=\tabspace
    \includegraphics[valign=c]{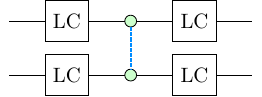}
\]
Applying such a 2-qubit entangling gate on a stabilizer code such as
\begin{equation}
    \label{eq:code_conversion_clifford}
    \includegraphics[valign=c]{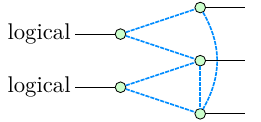}
    ~~~\longrightarrow~~~
    \includegraphics[valign=c]{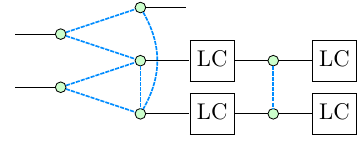}
    \hspace{1cm}
\end{equation}
has the effect of converting a given input stabilizer code into another output stabilizer code.
Depending on the architecture one may apply a final permutation of qubits as desired.

As an example consider the conversion sequence given in Figure \ref{fig:fowler_conv_sequence},
which fault-tolerantly converts from the \nkd{5}{1}{3} code \cite{laflamme_perfect_1996}
to the \nkd{7}{1}{3} Steane code \cite{steane_error_1996}.
\begin{figure}
    \centering
    \vspace*{-1cm}
    \hrulefill\\
    \vspace*{-1cm}
    \includegraphics[angle=-90,scale=0.38]{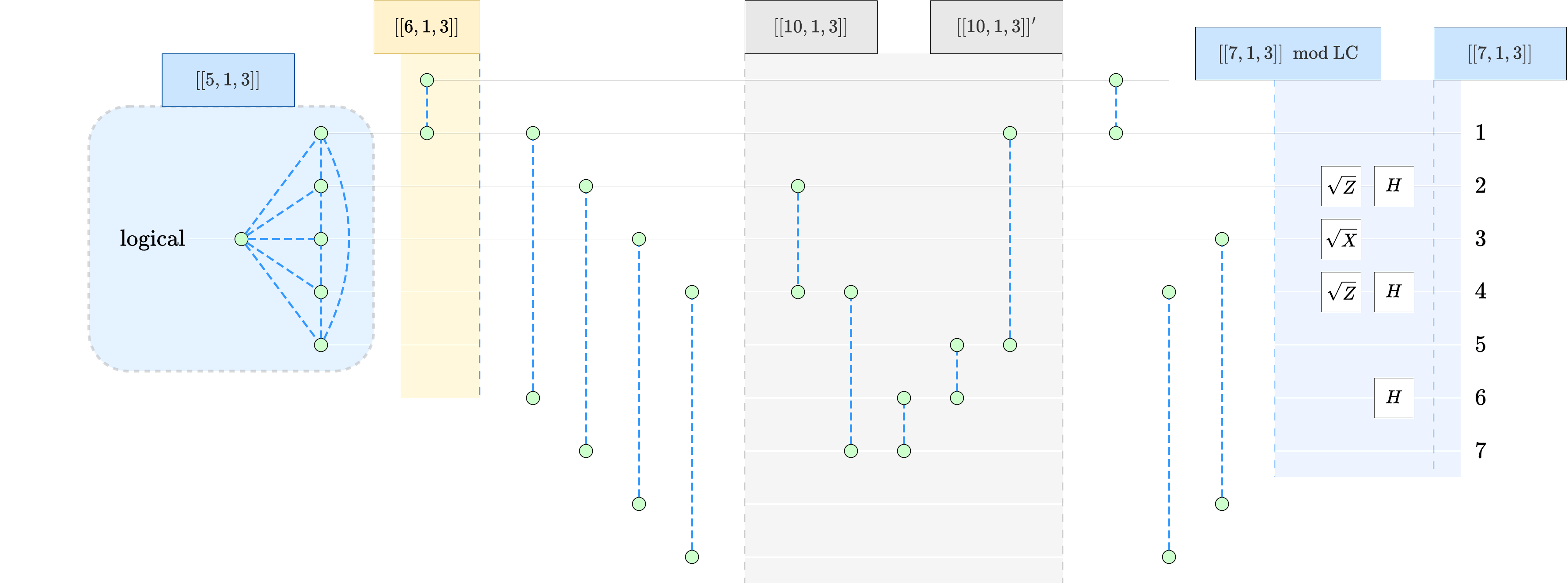}
    \vspace{\baselineskip/2}
    \caption{Clifford conversion sequence from the \nkd{5}{1}{3} code to the \nkd{7}{1}{3} Steane code. Found via breath-first search by Hill et al \cite{hill_fault-tolerant_2013} (equivalent version up to local Cliffords).}
    \label{fig:fowler_conv_sequence}
\end{figure}
We introduce the ZX encoding graph for the \nkd{5}{1}{3} code \cite{laflamme_perfect_1996}:
\begin{equation}
    \label{eq:513_graph_code}
    \includegraphics[valign=c]{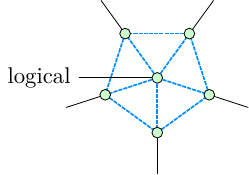}
\end{equation}
One can check that it corresponds to the expected stabilizers and logical operators by computing the stabilizer tableau from its vertex stabilizers as explained in the Letter.

For the resulting stabilizer code we note that one may freely reorder $\CZ$ gates, which we may thus pairwise cancel (Figure \ref{subfig:fowler_conv_cancel}).
Using the deliberate choice of local Cliffords we may then apply local complementation (Figure \ref{subfig:fowler_conv_LC}).
\begin{figure}
    \centering
    \subfloat[Pairwise cancelled $\CZ$ gates.]{
        \label{subfig:fowler_conv_cancel}
        \centering
        \begin{minipage}[t][5.75cm][c]{0.98\linewidth}
            \[
                \boxed{\hspace*{-5mm}\includegraphics[valign=c,scale=0.25]{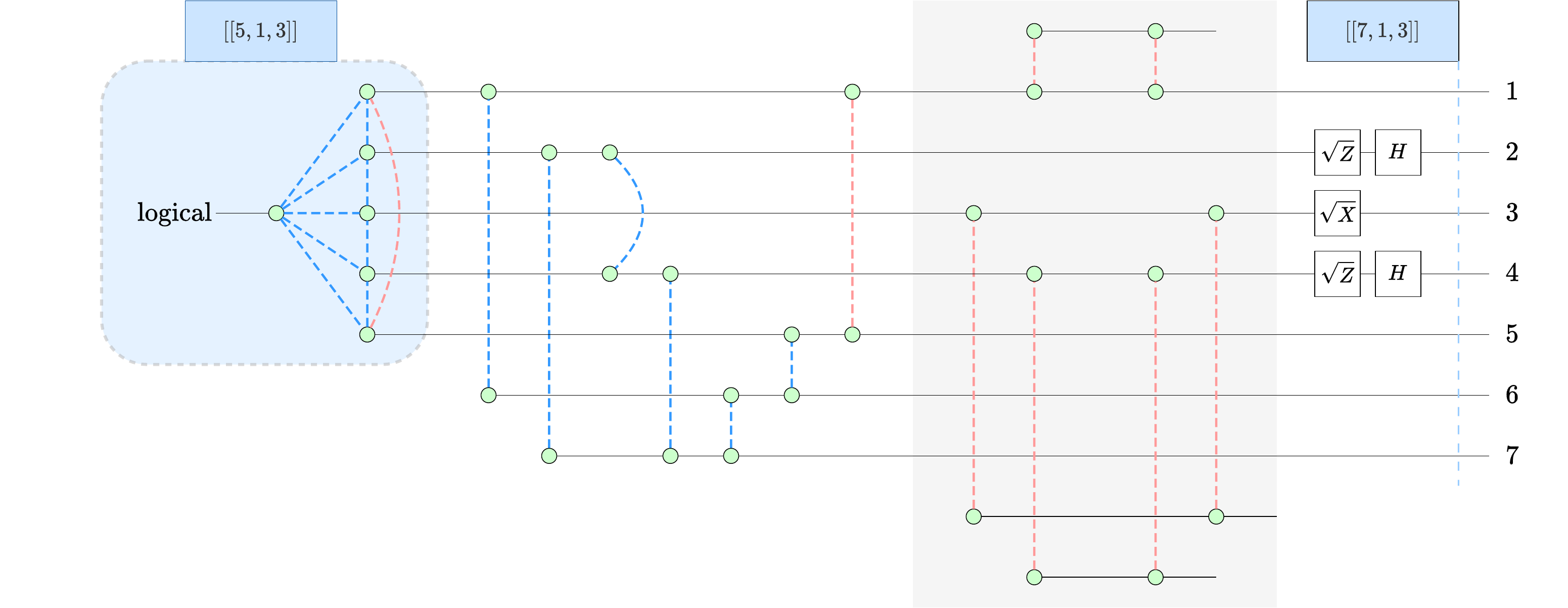}~=}
            \]
        \end{minipage}
        }
    \vspace{-1mm}
    \subfloat[Local complementation via deliberate choice of local Cliffords.]{
        \label{subfig:fowler_conv_LC}
        \centering
        \begin{minipage}[t][4.25cm][c]{0.9\linewidth}
            \[
                \boxed{=\hspace*{-1cm}\includegraphics[valign=c,scale=0.25]{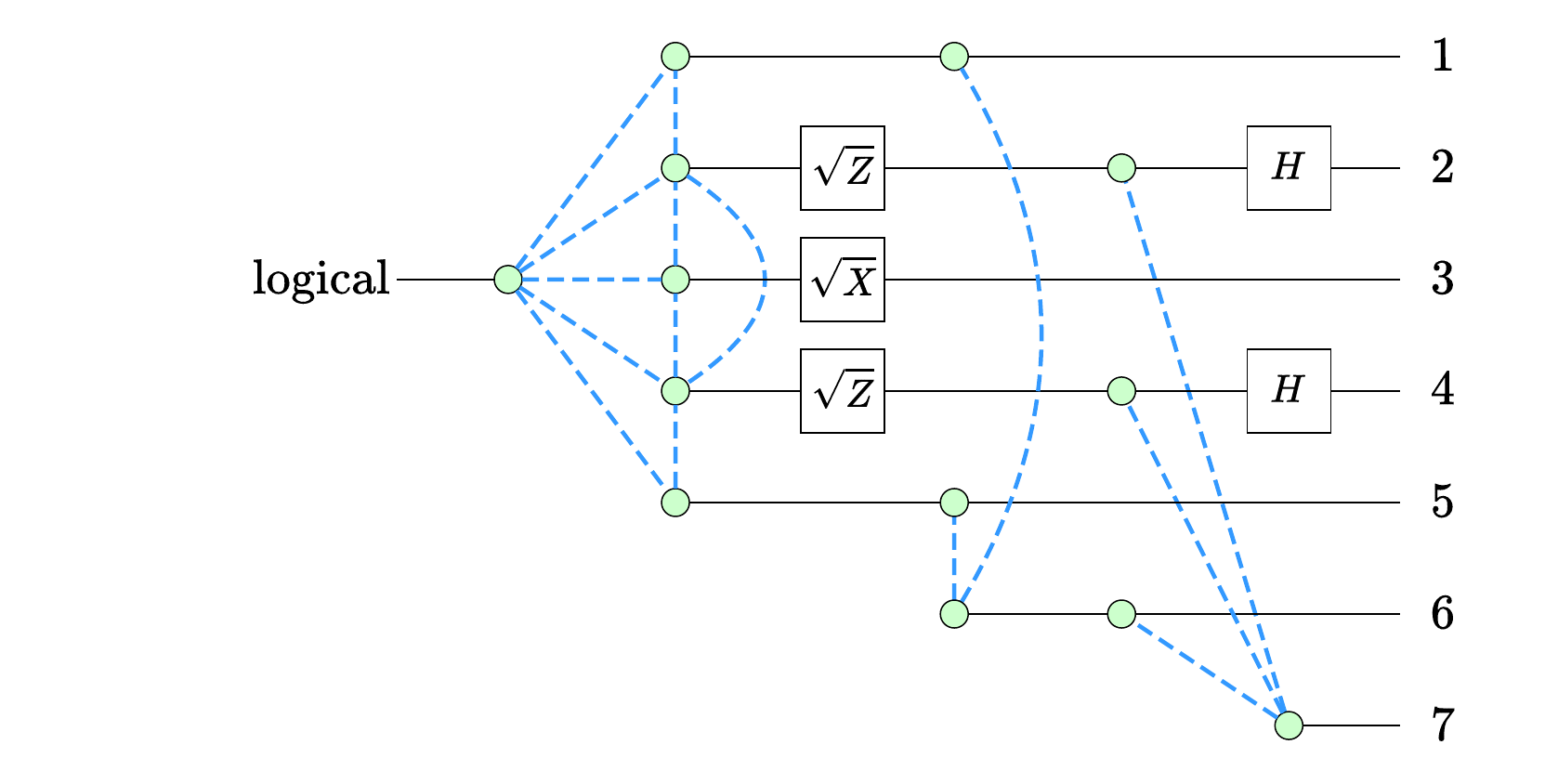}
                =\includegraphics[valign=c,scale=0.25]{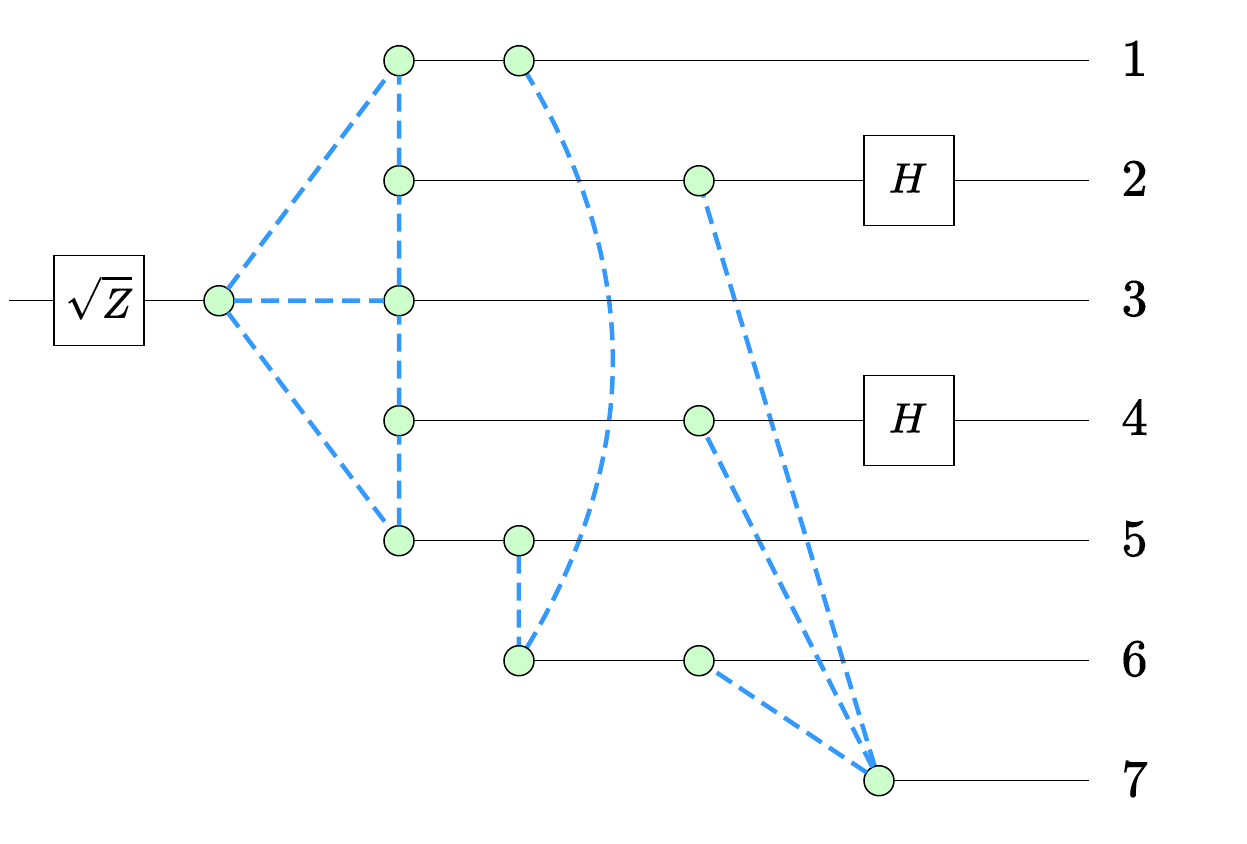}}
            \]
        \end{minipage}
        }
    \caption{
        How to easily read off the Steane code from the conversion sequence \eqref{fig:fowler_conv_sequence} via ZX-calculus:
        \eqref{subfig:fowler_conv_cancel} First, freely slide and cancel any pairwise $\CZ$ gates.
        \eqref{subfig:fowler_conv_LC} Second, apply local complementation (definition \ref{def:local-complementation}).
        Caution that rearranging $\CZ$ gates as in \eqref{subfig:fowler_conv_cancel} and \eqref{subfig:fowler_conv_LC} does not preserve the spacetime distance of the circuit \cite{rodatz_fault_2025}
        (and also not the code distance of the ISGs \cite{hastings_dynamically_2021}).
        For reading off the resulting stabilizer code this is however irrelevant.\\
    }
    \label{fig:fowler_conv_steane_code}
\end{figure}
The resulting stabilizer code is easily seen to resemble the ZX encoding graph for the Steane code (seen as a CSS code \cite{kissinger_phase_free_2022}):
\begin{equation}
    \label{eq:fowler_conv_steane_code}
    \eqref{subfig:fowler_conv_LC}
    =
    \includegraphics[valign=c]{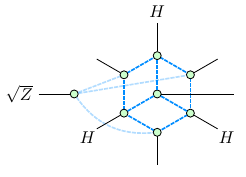}
    =
    \includegraphics[valign=c]{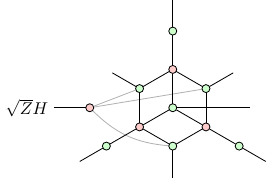}
\end{equation}
A similar conversion sequence from the \nkd{5}{1}{3} code to the Steane code was found in \cite{hwang_fault-tolerant_2015},
although it is not clear whether the suggested augmentation procedure provides fault tolerant measurement circuits.
Note also the ease with which we graphically read off the Steane code from the conversion sequence applied on the ZX encoding graph (see Figure \ref{fig:fowler_conv_steane_code} and \eqref{eq:fowler_conv_steane_code} compared with \cite{hill_fault-tolerant_2013}).

\subsection{Clifford Conversion \texorpdfstring{$\implies$}{=>} Gauge Fixing}\label{subsec:cliff_to_gauge}

We would like to relate now code conversion via Cliffords to code switching via gauge fixing.
Suppose for instance there was no single qubit Clifford present such as
\[
    \left(
    \includegraphics[valign=c]{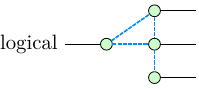}
    \right)
    \!
    \includegraphics[valign=c]{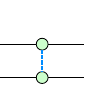}
    \tabspace = \tabspace
    \includegraphics[valign=c]{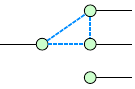}
\]
then this conversion arises equally via gauge fixing from the following subsystem (up to some final local Cliffords on physical qubits):
\begin{eqnarray}
    \label{eq:cliff_conv_subsystem}
    \includegraphics[valign=c]{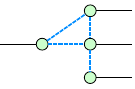}
    ~\xleftarrow{\includegraphics{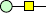}I}~
    \left(
    \includegraphics[valign=c]{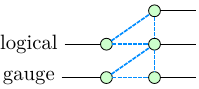}
    \right)
    ~\xrightarrow{\includegraphics{figures/snippets/gauge_state.pdf}\sqrt{X}}~
    \includegraphics[valign=c]{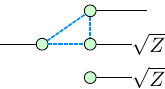}
\end{eqnarray}
More generally we however face an obstacle, as a single qubit Cliffords may block a direct identification via gauge fixing such as in example \eqref{eq:code_conversion_clifford} from above.
And indeed a Clifford conversion by CNOT has an entirely different effect as a binary adder of neighboring nodes (instead of toggling the connecting edge between qubits) such as
\begin{equation}
    \label{eq:binary_adder}
    \includegraphics[valign=c]{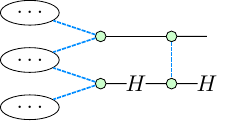}
    \tabspace = \tabspace
    \includegraphics[valign=c]{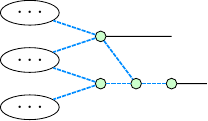}
    \tabspace = \tabspace
    \includegraphics[valign=c]{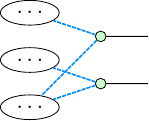}
\end{equation}
where we used deletion of interior edges as in \cite[lemma 5.3]{duncan_graph-theoretic_2020}.
Surprisingly, it is also possible to relate every such Clifford conversion to subsystem gauge fixing:

\begin{theorem}\label{theorem:clifford_conv_gauge_fixing}
    Every code conversion of stabilizer codes via a 2-qubit entangling Clifford (excluding swap of qubits)
    arises equally via subsystem gauge fixing (plus local Cliffords).\\
    More precisely, Clifford conversion arises equally as a single round of
    \begin{enumerate}
        \item measurement of a single gauge check
        \item followed by some local Cliffords on physical qubits
    \end{enumerate}
    including an identical effect on the logical encoding.
\end{theorem}

\begin{proof}
    Recall that Clifford conversion by a CZ gate has the identical effect on the ZX encoding graph as gauge fixing such as in \eqref{eq:cliff_conv_subsystem}
    \[
        \includegraphics[valign=c]{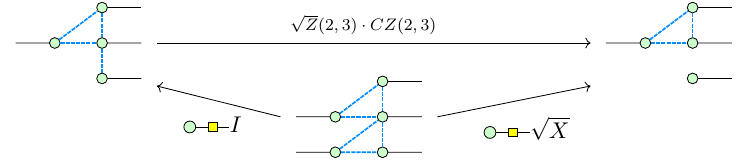}
    \]
    which among entails an identical transformation on the logical encoding (mod stabilizer).
    We thus consider an entangling Clifford (without a swap of qubits) split into an entangling CZ gate and some surrounding layer of local Cliffords,
    \[
        \includegraphics[valign=c]{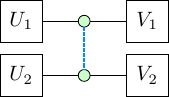}
        ~=~
        \left(
        \includegraphics[valign=c]{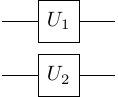}
        \right)
        \left(
        \includegraphics[valign=c]{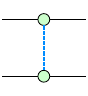}
        \right)
        \left(
        \includegraphics[valign=c]{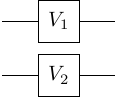}
        \right).
    \]
    This allows us to first consider the blocking layer of local Cliffords acting on a ZX encoding graph (representing the given a stabilizer code) such as
    \[
        \left(
        \includegraphics[valign=c]{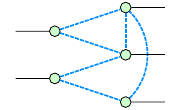}
        \right)
        \!\!
        \includegraphics[valign=c]{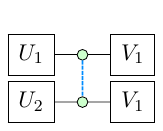}
        ~=~
        \left(
        \includegraphics[valign=c]{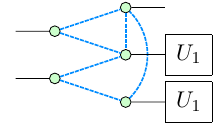}
        \:\right)
        \!\!\!\!\!\!
        \includegraphics[valign=c]{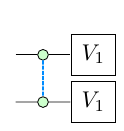}.
    \]
    We now aim to transform the ZX encoding graph
    until there is no more blocking layer of local Cliffords.
    For this we may first follow the same series of reductions via local complementation as in Section \ref{sec:subs_zx}
    such as the transformation (as a circuit identity)
    \[
        \includegraphics[valign=c]{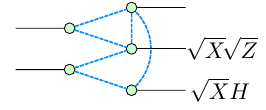}
        ~=~
        \includegraphics[valign=c]{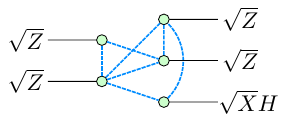}
        ~=~
        \includegraphics[valign=c]{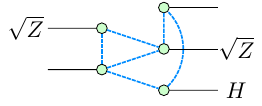}
    \]
    until the remaining local Cliffords are of the form $U_1,U_2 = I/\sqrt{Z}/H$.
    We may now postpone any remaining phase gate as a local Clifford past the entangling Clifford
    \[
        \includegraphics[valign=c]{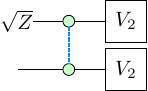}
        ~=~
        \includegraphics[valign=c]{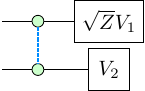}
        \tabspace\Bigg/\tabspace
        \includegraphics[valign=c]{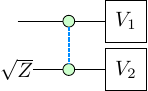}
        ~=~
        \includegraphics[valign=c]{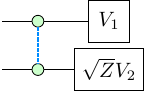}
    \]
    and as such we are left with at most some Hadamard gates, $U_1,U_2=I/H$.
    Suppose either the first or second qubit is adjacent to another node (either some logical or another physical qubit), then we may apply pivoting along these edges such as
    \[
        \includegraphics[valign=c]{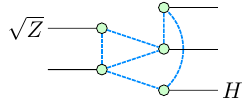}
        ~=~
        \includegraphics[valign=c]{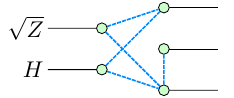}
    \]
    and thus move any blocking Hadamard out the way to another node as well.
    Similarly we may cancel any adjacent pair of Hadamards such as
    \[
        \includegraphics[valign=c]{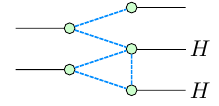}
        ~=~
        \includegraphics[valign=c]{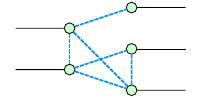}
    \]
    and we are thus left with no more blocking layer of local Cliffords.
    As such we may now realise the Clifford conversion by the entangling CZ gate via gauge fixing as in \eqref{eq:cliff_conv_subsystem}.

    Note that the previous removal of blocking Hadamard gates works only as long as there is at least one adjacent physical or logical qubit resp.~if the affected pair of qubits is connected.
    We are thus left with the following exceptional cases
    \begin{gather*}
        \includegraphics[valign=c]{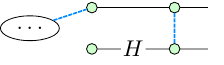}
        \tabspace\Bigg/\tabspace
        \includegraphics[valign=c]{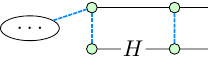}
        \tabspace\Bigg/\tabspace
        \includegraphics[valign=c]{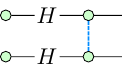}
    \end{gather*}
    where the dots $(\ldots)$ denote some potentially neighboring qubits.
    The first two perform a binary adder as in equation \eqref{eq:binary_adder} just without any neighboring nodes,
    and as such the Clifford conversion has no effect in here.
    Similarly the third one has no effect as may be seen from the following sequence of edge deletions \cite[lemma 5.3]{duncan_graph-theoretic_2020}
    \[
        \includegraphics[valign=c]{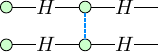}
        \tabspace = \tabspace
        \includegraphics[valign=c]{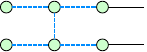}
        \tabspace = \tabspace
        \includegraphics[valign=c]{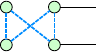}
        \tabspace = \tabspace
        \includegraphics[valign=c]{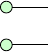}
    \]
    which also makes sense since this case is simply
    \[
        \Bigl(H\otimes H\mathop{CZ}H\otimes H\Bigr) \ket{+}\otimes\ket{+} = H\otimes H\mathop{CZ}\ket{0}\otimes\ket{0} = H\otimes H\ket{0}\otimes\ket{0} = \ket{+}\otimes\ket{+}.
    \]
    As such we end up with two cases:
    either we may clear out any local Clifford before the entangling CZ gate
    and apply gauge fixing with identical effect as in \eqref{eq:cliff_conv_subsystem},
    or the entire Clifford conversion has no effect on the given stabilizer code.

    Finally we note that conjugation by local Cliffords followed by gauge fixing followed by another round of local Cliffords
    has the same effect as gauge fixing and a single round of local Cliffords (either before or after gauge fixing).
    More precisely, the measurement of stabilizer and the conjugation by Cliffords commute in the sense
    \[
        \includegraphics{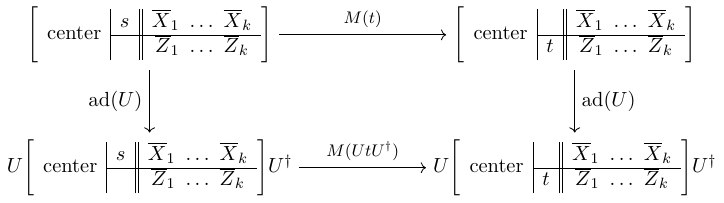}
    \]
    while noting that the commutation remains preserved under unitary conjugation.
    As such it suffices to perform gauge fixing followed by a single round of local Cliffords.
\end{proof}

So far we have found that Clifford conversion may be equally realized via gauge fixing.
Contrary to code deformation and gauge fixing being equivalent \cite{vuillot_code_2019},
we will however find obstructions for Clifford conversion based on weight enumerators,
and thereby provide an example of gauge fixing which does not arise via any single 2-qubit entangling Clifford.
As such, we find that gauge fixing is strictly more general than conversion via merely single 2-qubit entangling Cliffords.
But before we do so, let us first unravel how gauge fixing may be understood as conversion by some highly entangling Clifford in the following subsection.

\subsection{$\pi/4$ rotations: Highly entangling}\label{subsec:gauge_entang}

It is well-known that the transformation by measuring gauge checks arises equally as Clifford unitaries
\cite{colladay_rewiring_2018,huang_transversal_2018,aasen_measurement_2023}
which define symplectic transvections
\cite{koenig_how_2014,rengaswamy_synthesis_2018,rengaswamy_logical_2020,chen_fault_2025}:
\vspace{-\baselineskip/2}
\[
    \includegraphics{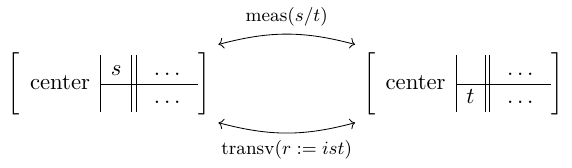}
    \vspace{-\baselineskip/2}
\]
Symplectic transvections thereby act the same as an abstract Hadamard gate:
they swap the current stabilizer and destabilizer pair (as observed in \cite{colladay_rewiring_2018,huang_transversal_2018,aasen_measurement_2023})
\[
    \transv(r:=ist) = \ad\!\left(\frac{1+ir}{\sqrt{2}}\right):\tabbspace s\mapsto t,\tabspace t\mapsto -s
\]
while leaving the remaining stabilizer tableau untouched.
This relation was identically rediscovered in \cite[prop.~2.2]{kliuchnikov_stab_circ_2023},
however identified as an evident $\pi/4$ rotation:
\[
    \exp\!\left(\frac{i\pi}{4} r\right) = \cos(\pi/4) + ir\sin(\pi/4) = \frac{1+ir}{\sqrt{2}}
\]
We would like to view this however as phase type gates.
For this we may assume that the stabilizer defining the $\pi/4$ rotation arises as a Pauli Z string (up to local Cliffords)
\[
    r = (P_1\ldots P_m)\otimes (I\ldots I) \stackrel{\LC}{\sim} (Z\ldots Z)\otimes (I\ldots I) = Z(\supp).
\]
The resulting $\pi/4$ rotation arises then as a fully connected set of $\CZ$ gates:
\[
    \exp\left(-\frac{i\pi}{4} Z(\supp)\right) = \sqrt{Z}(\supp)\Biggl(\,\prod_{x,y\in\supp}CZ(x,y)\Biggr)
\]
This may be seen most easily on the Pauli basis: that is for every $X(a\in\supp)$,
\begin{gather*}
    \exp\left(-\frac{i\pi}{4} Z(\supp)\right)X(a)\exp\left(i\frac{\pi}{4} Z(\supp)\right) = X(a)\exp\left(\frac{i\pi}{2} Z(\supp)\right)
    = iX(a)Z(\supp)\\[\jot]
    =\sqrt{Z}(\supp)\Biggl(\,\prod_{x,y\in\supp}CZ(x,y)\Biggr)X(a) \Biggl(\,\prod_{x,y\in\supp}CZ(x,y)\Biggr) \sqrt{Z}(\supp)
    = Y(a)Z(\supp{\setminus} a) \tabspace (\text{mod signs})
\end{gather*}
and with trivial action on every $X(b\notin\supp)$ and on every $Z(a)$.
As such, $\pi/4$ rotations define highly entangling Cliffords, such as even for mere weight-6 rotation
\[
    \exp
    \!\left(\!
    \includegraphics[valign=c]{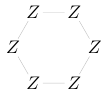}
    \!\right)
    =
    \left(
    \includegraphics[valign=c]{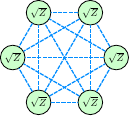}
    \right)
    = \left(~
    \includegraphics[valign=c]{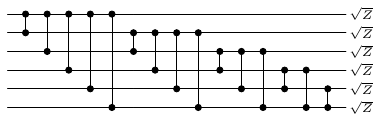}
    ~\right).
\]
Even for a mere weight-4 rotation we already obtain
\begin{gather*}
    \exp\left(\frac{i\pi}{4} ZZZZ\right) = \sqrt{Z}\!\otimes\!\sqrt{Z}\!\otimes\!\sqrt{Z}\!\otimes\!\sqrt{Z}\,\cdot\, CZ(12)CZ(13)\cdots CZ(34)
\end{gather*}
such as while moving surface code defects during code deformation as in Figure \ref{subfig:moving-defects}.
Therefore such rotations are hard to realize directly via individual fault-tolerant 2-qubit entangling Cliffords
\footnote{Whereas such $\pi/4$ rotations are easy to realize fault-tolerantly via lattice surgery by consuming an ancilla.}.
Note that the actual $\pi/4$ rotation arises by the \enquote{difference} between gauge checks
\[
    \left[\begin{array}{c|c||c}
    \multirow{2}{*}{$\centerstab$}
        & s & ~\ldots~ \\
        \cline{2-3}
        && ~\ldots~ \\
    \end{array}\right]
    \longrightarrow
    \left[\begin{array}{c|c||c}
    \multirow{2}{*}{$\centerstab$}
        && ~\ldots~ \\
        \cline{2-3}
        & t & ~\ldots~ \\
    \end{array}\right]:
    \tabbspace
    \meas(t)
    =\rot(r,\pi/4)
    \tabspace
    \text{with}
    \tabspace
    r=ist
\]
where $\rot(P,\phi)=\ad(e^{i\phi P})$ denotes rotation around a selfadjoint Pauli,
and thus the same issue arises for any equivalent $\pi/4$ rotation around any rotation axis mod center (generally just worse)
\[
    \meas(t)=\meas(t\cdot\centerstab)=\rot(r,\pi/4)
    \tabspace
    \text{for}
    \tabspace
    r=r^\dagger=(ist)\cdot\centerstab
\]
as typically the originally chosen gauge check already comes with minimal weight.
Further note that each such \mbox{$\pi/4$ rotation} is not implementable as a single Clifford conversion step
(because of overlapping $\CZ$ gates) such as for a weight-4 stabilizer
\[
    \includegraphics[valign=c]{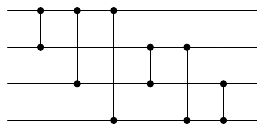}
    =
    \includegraphics[valign=c]{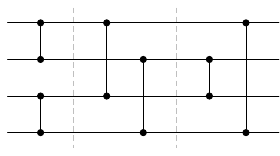}
\]
and thus realizing such $\pi/4$ rotations as a circuit of 2-qubit entangling Cliffords
(even when applying non-overlapping gates simultaneously)
leads to intermediate stabilizer groups (ISGs) with generally lower code distances.
In addition, under standard 2-qubit depolarizing noise,
some of the 2-qubit entangling Clifford gates may produce harmful \enquote{non-transversal} faults in-between ISGs,
and thus effectively lower the fault distance \cite{rodatz_fault_2025}
(although see \cite{chen_fault_2025} on the fault tolerance of rotations considered as symplectic transvections).

Such $\pi/4$ rotations however are not the only Cliffords implementing gauge check measurements.
More generally, note that a single gauge check measurement implements the transformation
\begin{gather*}
    \renewcommand{\arraystretch}{1.2}
    \left[\begin{array}{c|c||c}
        \multirow{2}{*}{$\centerstab$}
        & s & \logicalX_1~\ldots~\logicalX_k \\
        \cline{2-3}
        && \logicalZ_1~\ldots~\logicalZ_k \\
    \end{array}\right]
    \xrightarrow{\text{measure:}~t}
    \left[\begin{array}{c|c||c}
        \multirow{2}{*}{$\centerstab$}
        && \logicalX_1~\ldots~\logicalX_k \\
        \cline{2-3}
        & t & \logicalZ_1~\ldots~\logicalZ_k \\
    \end{array}\right]
\end{gather*}
which is determined on gauge stabilizer and logical Paulis \emph{only up to center stabilizer},
\[
    s\mapsto t
    \tabbspace
    t\mapsto (s~\mathor~ st)
    \tabbspace
    \logicalX_i\mapsto \logicalX_i
    \tabbspace
    \logicalZ_i\mapsto \logicalZ_i
    \tabbspace(\mod~\centerstab)
\]
while on the center by some completely abritrary Clifford transformation $\centerstab\to\centerstab$.
More visually, any Clifford implementing a single gauge check is thus of the form
\[
    \left(\begin{array}{c|c|ccc}
        U_\centerstab\bigstrut & U_{st} & U_{1} & \cdots & U_{k} \\
        \hline
        0 & \begin{smallmatrix}0&1\\1&*\end{smallmatrix}\bigstrut & 0 & 0 & 0 \\
        \hline
        0 & 0 & \begin{smallmatrix}1&0\\0&1\end{smallmatrix}\bigstrut & 0 & 0 \\
        0 & 0 & 0 & \smash{\ddots}\bigstrut & 0 \\
        0 & 0 & 0 & 0 &
        \begin{smallmatrix}1&0\\0&1\end{smallmatrix}\bigstrut
    \end{array}\right)
\]
where we used an algebraic tensor product decomposition (for better visualization):
\[
    \left[\begin{array}{c|c||c}
        \multirow{2}{*}{$s_1~\ldots~s_m$}
        & s & \logicalX_1~\ldots~\logicalX_k \\
        \cline{2-3}
        & t & \logicalZ_1~\ldots~\logicalZ_k \\
    \end{array}\right]
    =
    \left[\begin{array}{c|c||c}
        \multirow{2}{*}{$X_{-m}~\ldots~X_{-1}$}
        & X_0 & X_1~\ldots~X_k \\
        \cline{2-3}
        & Z_0 & Z_1~\ldots~Z_k \\
    \end{array}\right].
\]
The number of such Cliffords (besides transvections as discussed above) is however way beyond what is computable.
For instance even the number of Cliffords implementing a single gauge check measurement between a pair of \nkd{7}{1}{3} stabilizer codes is already
\begin{align*}
    &
    \underbrace{\bigl(1\times2^m\bigr)\cdot\bigl(2\times2^m\bigr)}_{s\mapsto t,~t\mapsto s/t~(\mod~\centerstab)} \cdot
    \underbrace{\bigl(2^m\times\ldots\times 2^m\bigr)}_{X_i\mapsto X_i~(\mod~\centerstab)} \cdot
    \underbrace{\bigl(2^m\times\ldots\times 2^m\bigr)}_{Z_i\mapsto Z_i~(\mod~\centerstab)} \cdot
    \underbrace{|\symplecticgroup(2m)|\bigstrut}_{\centerstab\to\centerstab} \\
    &= \left(2\cdot2^{2m+2km}\right)\left(2^{m^2}\prod_{l=1}^m(4^l-1)\right)
    \geq \left(2\cdot 2^{2m+2km}\right)\cdot
    \left(\frac{3}{4}\right)^{\!\!1/3}\!\left(2^{m^2}\cdot2^{m(m+1)}\right)
    \overunderset{m=5}{k=1}{=\joinrel=}\:6^{1/3}2^{75}
\end{align*}
where we used the lower bound given by the convexity of the logarithm
\[
    \ln\!\left(\prod_{k=1}^\infty(1-q^k)\right)
    = \sum_{k=1}^\infty\ln\!\left(1-q^k\right)
    \geq \ln(1-q)\!\left(\sum_{k=1}^\infty q^k\right)
    = \ln(1-q)\frac{q}{1-q}
    \overset{q=1/4}{=\joinrel=}
    \frac{\ln(3/4)}{3}
\]
and the well-known order formula for the symplectic group (see for instance \cite{gross_hudsons_2006}).
As such we leave it as a question for further research on low overhead, fault-tolerant Cliffords.

\subsection{Code Deformation \texorpdfstring{$>$}{>} Clifford Conversion}\label{subsec:code_deform}

\begin{figure}
    \centering
    \subfloat[]{
        \label{subfig:moving-defects}
        \centering
        \includegraphics[valign=c,angle=-90,scale=1.5]{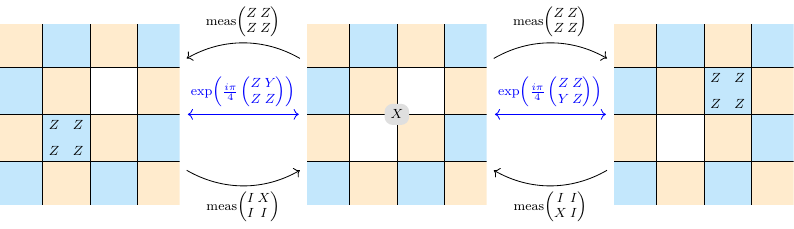}
        } 
    \hspace{1.5cm}
    \subfloat[]{
        \label{subfig:switching-defects}
        \centering
        \includegraphics[valign=c,angle=-90,scale=1.6]{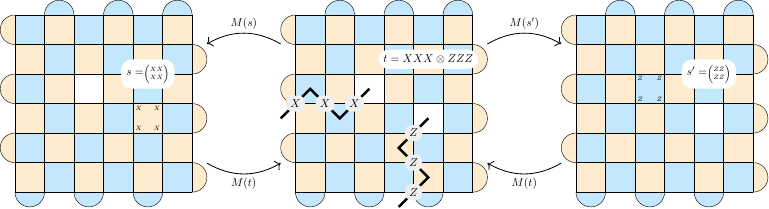}
        }
    \caption{Code deformations based on \cite{colladay_rewiring_2018}. (a) Moving surface code defects may be realized via gauge fixing or equivalently via $\pi/4$ rotations. This  example was taken from \cite[Figure 1]{colladay_rewiring_2018}. (b) Switching between $e$- and $m$-type defects may be realized via condensation at boundaries (center patch). This example was extended from \cite[Figure 2]{colladay_rewiring_2018}.}
    \label{fig:code-deformations}
\end{figure}

We now give an example of gauge fixing via single gauge check which cannot be implemented as single Clifford conversion step.
Consider for this the code deformation in Figure \eqref{subfig:switching-defects}.
\linebreak
Note that the example given in \cite[Figure 2]{colladay_rewiring_2018}
requires condensation at boundaries \cite{dennis_topological_2002,lidar_topological_2013}
which we added accordingly in Figure \eqref{subfig:switching-defects}.
We now claim that either direction of these single gauge check measurements (in this code deformation example)
cannot be realised as conversion via any single 2-qubit entangling Clifford as in equation \eqref{eq:code_conversion_clifford}
\[
    \includegraphics[valign=c]{figures/cliff_conv/cliff_conv_1.pdf}
    ~~~\longrightarrow~~~
    \includegraphics[valign=c]{figures/cliff_conv/cliff_conv_2.pdf}
    \hspace{1cm}
\]
and thus defines an example strictly beyond Clifford conversion as considered in \cite{hill_fault-tolerant_2013,hwang_fault-tolerant_2015}.\linebreak
For this note that 2-qubit Cliffords modify stabilizer weights by at most $\pm1$,
and as such can modify weight enumerators at most locally in the sense
\begin{align*}
    U\in\cliffords(2):\tabspace&\enum(S) := (\ldots,w_{i-1},w_{i},w_{i+1},\ldots)\\
    &\to \enum(USU^\dagger) = (\ldots,w_{i-1}+a,w_{i}-(a+b),w_{i+1}+b,\ldots)
\end{align*}
where $\enum(S)=(w_0,\ldots,w_n)$ denotes the weight enumerator for a stabilizer group, i.e., a list of the number of stabilizers of weight $0,\ldots, n$.
Note however that every stabilizer (including any product of plaquette stabilizer) is of even weight for the stabilizer codes in \eqref{subfig:moving-defects}.
As such their weight enumerators read:
\begin{align*}
    \enum(\LHS) &= \Bigl(\begin{matrix}1&0&12&0&w_4&0&\ldots&0&w_{48}&0\end{matrix}\Bigr) \\
    \enum(\MHS) &= \Bigl(\begin{matrix}1&0&12&0&w_4-1&0&\ldots&0&w'_{48}&0\end{matrix}\Bigr)\\
    \enum(\RHS) &= \Bigl(\begin{matrix}1&0&12&0&w_4&0&\ldots&0&w_{48}&0\end{matrix}\Bigr)
\end{align*}
Such a transformation is impossible via any single 2-qubit entangling Clifford!
Combining Theorem \ref{theorem:clifford_conv_gauge_fixing} with this example, we thus conclude:
gauge fixing is strictly more general than code conversion via (a sequence of) 2-qubit entangling Cliffords.
This lies in contrast to code deformation and gauge fixing being equivalent \cite{vuillot_code_2019}.

In particular, Clifford conversion as considered in \cite{hill_fault-tolerant_2013,hwang_fault-tolerant_2015} may be equally captured via gauge fixing, however not vice versa.
In a forthcoming paper, we will uncover further gauge fixing examples that go beyond Clifford conversion between every \nkd{7}{1}{3} code.

\end{document}